%% file: OnProductCodes.tex
\documentclass[10pt,twoside,a4paper]{amsart}

\input{Setting}
\input{Notation}

\begin{document}

\title{On tensor products of CSS Codes}
\author[B. Audoux]{Benjamin Audoux}
         \address{Aix Marseille Universit\'e, I2M, UMR 7373, 13453 Marseille, France}
         \email{benjamin.audoux@univ-amu.fr}
\author[A. Couvreur]{Alain Couvreur}
         \address{INRIA \& LIX, UMR 7161, École Polytechnique, 91128 Palaiseau, France}
         \email{alain.couvreur@lix.polytechnique.fr}
\date{\today}

\begin{abstract}
  CSS codes are in one-to-one correspondance with length 3 chain
  complexes. The latter are naturally endowed with a tensor product
  $\otimes$ which induces a similar operation on the former. We
  investigate this operation, and in particular its behavior with
  regard to minimum distances.  Given a CSS code $\qcode{C}$, we give
  a criterion which provides a lower bound on the minimum distance of
  $\qcode{C} \otimes \qcode{D}$ for every CSS code $\qcode D$.
  From this criterion arises a generic bound for the minimum
    distance which is twice larger
    than the single bound previously known in the literature.  We
  apply these results to study the behaviour of iterated tensor powers
  of codes.  Such sequences of codes are logarithmically LDPC and we
  prove in particular that their minimum distances tend generically to
  infinity. More precisely, their minimum distance increases as $O(n^\alpha)$
  for some $\alpha>0$, where $n$ is the code length,
  while the row weight of
    their parity--check matrices grows as $O(\log(n))$. This entails a rather surprizing fact:
    even if a CSS code does not have quantum degeneracy, for a large enough $\ell$, its
    $\ell$--th iterated tensor power does.
  Different known results are also reinterpretated in terms of tensor products and three new families of LDPC CSS codes are studied.
\end{abstract}

\maketitle

\section*{Introduction}
\input{0.1.Introduction}

\subsection*{Organization}
\input{0.2.Organization}

\subsection*{Notation}
\input{0.3.Notation}

 \subsection*{Acknowledgements}
 \input{0.4.Acknowledgement}

%%%%%%%%%%%%%%%%%%%%%%%%%%%%%%%%%%%%%

\section{Some background}
\label{sec:Background}

\subsection{Chain complexes}
\label{sec:ChainComplexes}
\input{1.1.ChainComplexes}

\subsection{Classical codes}
\label{sec:ClassicalCodes}
\input{1.2.ClassicalCodes}

\subsection{CSS codes}
\label{sec:CSS}
\input{1.3.CSSCodes}

%%%%%%%%%%%%%%%%%%%%%%%%%%%%%%%%%%%%%

\section{Tensor products of CSS codes}
\label{sec:Main}

\subsection{Definitions}
\label{sec:CSSpowers}
\input{2.1.CSSProduct}

\subsection{A minimum distance result}

\label{sec:Th}
\input{2.2.Lemma}

\subsection{Some direct consequences}
\label{sec:Consequences}
\input{2.3.DirectConsequences}

%%%%%%%%%%%%%%%%%%%%%%%%%%%%%%%%%%%%%

\section{Reinterpretation of known results}
\label{sec:KnownResults}

\subsection{Tillich--Zemor codes}
\label{sec:TZ}
\input{3.1.TillichZemor}

\subsection{Khovanov codes}
\label{sec:Khovanov}
\input{3.2.Khovanov}

\subsection{Product of Steane $\llbracket7;1;3\rrbracket$ codes}
\label{sec:Steane}
\input{3.3.Steane}

\subsection{Bravyi--Hastings homological product}
\label{sec:BH}
\input{3.4.BravyiHastings}

%%%%%%%%%%%%%%%%%%%%%%%%%%%%%%%%%%%%%

\section{New families of codes}
\label{sec:NewFamilies}
\input{4.0.Intro_new_families}

\subsection{Quantum finite geometry codes}
\label{sec:QFG}
\input{4.2.QFG}

\subsection{Quantum cyclic codes}
\label{sec:QCyclic}
\input{4.3.QCyclic}

\subsection{Quantum Reed--Muller codes}
\label{sec:QRM}
\input{4.1.QRM}

%%%%%%%%%%%%%%%%%%%%%%%%%%%%%%%%%%%%%

%%%%%%%%%%%%%%%%%%%%%%%%%%%%%%%%%%%%%

% Bibliographie
\bibliographystyle{amsalpha}
\bibliography{OnProductCodes}
\addcontentsline{toc}{part}{Bibliography}

\appendix

\section{Length of tensor powers}
\label{appendix:A}
\input{A.LengthPowers}

\section{Length of reduced tensor powers}
\label{appendix:B}
\input{B.LengthRedPowers}

\end{document}

%% file: Setting.tex
\usepackage[english]{babel}
\usepackage[utf8]{inputenc}

\usepackage{amsfonts,amsthm,amssymb,amsmath,amsthm,latexsym,wasysym,,mathrsfs,dsfont,txfonts}
\usepackage{accents,multirow,rotating,subfigure,graphicx,bigstrut,lscape,multicol,fancyhdr,enumerate,accents,mathtools,exscale}

\usepackage{pdftricks}
\usepackage{color}
\usepackage[pdftex,all]{xy}

\usepackage{hyperref}
\usepackage{appendix}
\usepackage[justification=centering]{caption}

\graphicspath{{Figures/}}

\theoremstyle{theorem}
\newtheorem {theo}{Theorem}[section]
\newtheorem*{theo*}{Theorem}
\newtheorem {lemme}[theo]{Lemma}
\newtheorem*{lemme*}{Lemma}
\newtheorem {prop}[theo]{Proposition}
\newtheorem*{prop*}{Proposition}
\newtheorem {cor}[theo]{Corollary}
\newtheorem*{cor*}{Corollary}

\newtheorem*{cor_proof*}{Corollary (of the proof)}

\newtheorem*{conjecture*}{Conjecture}
\theoremstyle{definition}
\newtheorem {defi}[theo]{Definition}
\newtheorem*{defi*}{Definition}
\newtheorem {nota}[theo]{Notation}
\newtheorem*{nota*}{Notation}
\theoremstyle{remark}
\newtheorem {remarque}[theo]{Remark}
\newtheorem*{remarque*}{Remark}

\newtheorem*{Iremarque*}{Important remark}
\newtheorem {warning}[theo]{Warning}
\newtheorem*{warning*}{Warning}

\newtheorem*{remarques*}{Remarks}

\newtheorem*{warnings*}{Warnings}

\newtheorem*{convention*}{Convention}
\newtheorem {exemple}[theo]{Example}
\newtheorem*{exemple*}{Example}

\newtheorem*{exemples*}{Examples}

\newtheorem*{question*}{Question}

\newtheorem*{questions*}{Questions}

\newtheorem*{fact*}{Fact}

%% file: Notation.tex
\def\BB{{\mathcal B}}

\def\CC{{\mathcal C}}

\def\F{{\mathds F}_2}
\def\Fq{{\mathds F}_q}

\def\N{{\mathds N}}

\def\Z{{\mathds Z}}

\def\2Z{{\fract{\Z}/{2\Z}}}

\def\e{\varepsilon}
\def\p{\partial}

\def\Ch{{\textnormal{Ch}}}

\def\Hom{{\textnormal{Hom}}}
\def\Id{{\textnormal{Id}}}
\def\Im{{\textnormal{Im}}}

\def\Ker{{\textnormal{Ker}}}
\def\Coker{{\textnormal{Coker}}}
\def\Mat{{\textnormal{Mat}}}
\def\End{{\textnormal{End}}}

\def\deg{{\textnormal{deg}}}

\def\overlap{{\textnormal{overlap}}}

\def\pcup{\operatornamewithlimits{\cup}\limits}

\def\poplus{\operatornamewithlimits{\oplus}\limits}

\def\potimes{\operatornamewithlimits{\otimes}\limits}
\def\pprod{\operatornamewithlimits{\prod}\limits}

\def\psum{\operatornamewithlimits{\sum}\limits}
\def\pmin{\operatornamewithlimits{\min}\limits}

\def\fract#1/#2{\hbox{\leavevmode
  \kern.1em \raise .25ex \hbox{\the\scriptfont0 $#1$}\kern-.1em }\big/
  {\hbox{\kern-.15em \lower .5ex \hbox{\the\scriptfont0 $#2$}} }}

\def\ffract#1/#2{\hbox{\leavevmode
  \kern.1em \raise .25ex \hbox{\the\scriptfont0 $#1$}\kern-.1em }\big/
  {\hbox{\kern-.15em \lower .5ex \hbox{\the\scriptfont0 \scriptsize $#2$}} }}

\def\fractt#1/#2{\hbox{\leavevmode
  \kern.1em \raise .25ex \hbox{\the\scriptfont0 $#1$}\kern-.1em
}\lower .2ex\hbox{\Big/}
  {\hbox{\kern-.15em \lower .8ex \hbox{\the\scriptfont0 $#2$}} }}

\def\subfract#1/#2{\hbox{\leavevmode
  \kern.1em \raise .25ex \hbox{\the\scriptfont0 \scriptsize $#1$}\kern-.1em }/
  {\hbox{\kern-.15em \lower .5ex \hbox{\the\scriptfont0 \scriptsize $#2$}} }}

\newcommand{\smallfrac}[2]{\textrm{\the\scriptfont0 \scriptsize $\frac{#1}{#2}$}}

\newcommand{\noi}{\noindent}
\newcommand{\disp}{\displaystyle}

\newcommand{\dessin}[2]{
  \vcenter{\hbox{\includegraphics[height=#1]{#2.pdf}}}}

\newcommand{\function}[5]{
  #1 \colon \left\{
  \begin{array}{ccc}
    #2 & \longrightarrow & #4\\[.1cm]
    #3 & \longmapsto & #5
  \end{array}\right.}
\newcommand{\fonction}[4]{
  \begin{array}{ccc}
    #1 & \longrightarrow & #3\\[.1cm]
    #2 & \longmapsto & #4
  \end{array}}

\newcommand{\Inter}[2]{\{#1,\ldots,#2\}}

\newcommand{\param}[4]{{\left\llbracket#1,#2,#3,#4\right\rrbracket}}

\def\RM{\textrm{RM}}
\def\QRM{\textrm{QRM}}
\def\QFG{\textrm{QFG}}
\def\QCC{\textrm{QCC}}
\def\Pol{\textrm{Pol}}
\def\Span{\textrm{Span}}

\newcommand{\chain}[1]{\mathscr{#1}}
\newcommand{\qcode}[1]{\mathcal{#1}}
\DeclareMathAlphabet{\mathpzc}{OT1}{pzc}{m}{it}
\newcommand{\code}[1]{\mathpzc{#1}}

\def\cperp{{\mathrlap{\hspace{.02cm}\perp}\bigcirc}}
\def\rotimes{\otimes_r}
\def\protimes{\operatornamewithlimits{\rotimes}\limits}

\renewcommand{\P}{\mathbb{P}}
\renewcommand{\leq}{\leqslant}
\renewcommand{\geq}{\geqslant}
\newcommand{\clines}{\code{C}_{{\rm lines}}}
\newcommand{\cplanes}{\code{C}_{{\rm planes}}}
\newcommand{\one}{\mathbf{1}}

%% file: 0.1.Introduction.tex
In the last century, error-correcting codes were developed to overcome
the emergence of anomalies in data while transmitting or storing
them. In the quantum setting, such correction systems are all the more
important as quantum decoherence eventually produces such errors. At
the end of the $\textrm{XX}^\textrm{th}$ century, several
constructions were given for quantum error-correcting codes; among
them, CSS codes, developped by A.R Calderbank, P. Shor and A . Steane
\cite{Calderbank,Steane}, are constructed from two classical codes
orthogonal to each other.  Because of their strong relation with
classical codes, they have been the subject of intense study.  CSS
codes can alternatively be related to the topological notion of chain
complexes. Not only does this approach provides a way to construct CSS
codes, but parameters such as length, dimension and minimum distance
can also be read through the chain complex and its (co)homology.  From
this perspective, the non detectable error patterns correspond to
(co)cycles belonging to nonzero classes in the (co)homology of the chain complex.
The dimension of the quantum code is nothing but the dimension of the
(co)homology group, and the quantum minimum
distance nothing but the minimum weight of a (co)homologically non
trivial cycle. This point of view was pioneered by M. Freedman,
D. Meyer \cite{Meyer} and A. Kitaev \cite{Kitaev}.

Swiftness in error-correction is crucial since error correction should
occur faster than errors arise.  In the classical setting, LDPC (Low
Density Parity Check) codes, that is codes with sparse parity-check
matrices \cite{Gallager} are known to have very efficient decoding
algorithms.  These so-called {\em iterative decoding algorithms} have
a very low complexity and can be applied to some LDPC codes with error
correction performance very close to the Shannon limit; see for
instance \cite{RiUrSh}.  The notion of LDPC code can be transposed to
the quantum setting and efficient iterative decoding algorithms exist
\cite{LTZ}. However, while good classical LDPC codes can easily be
obtained by random generation, there is no way to generate randomly
quantum LDPC codes. Hence, to date, the construction of a quantum LDPC
code rests on methods of algebraic topology
\cite{Kitaev,Bombin,Zemor,Audoux} or combinatorics
\cite{TZ,Couvreur,Nico}. The list of references is far from being
exhaustive.  Classical and quantum LDPC codes differ in yet another
important point. While a generic sequence of classical LDPC codes has
a minimum distance which is linear in the code length, the best known
families of quantum LDPC codes have a minimum distance in
$O \left(\!\sqrt{n\sqrt{\log n}}\right)$ \cite{FreedmanLuo}.  The
question whether this square root barrier is fortuitous or not remains
open.  It is worth noting that by ``LDPC'' we mean that the code has
parity check matrices with row weight in $O(1)$ or $O(\log n)$ where
$n$ denotes the code length.  Indeed, a recent result of Bravyi and
Hastings \cite{Hastings} proves the existence of MDPC (Moderate
Density Parity Check) quantum codes, that is CSS codes described by
matrices whose row weight is in $O (\! \sqrt{n})$, with dimensions and
minimum distances linear in the code length.  Bravyi and Hastings'
construction is performed by choosing two 
CSS codes of the same length and by computing their so-called {\em
  homological product}, that we shall denote here by $\boxtimes$. It
produces a
CSS code whose minimum distance is linear in its length with a nonzero
probability.

In the present paper, we deepen the interplay between chain complexes
and CSS codes by transposing to the latter the standard notion of
tensor product $\otimes$ defined for the former.  We also introduce a
reduced notion $\rotimes$ of tensor product which, compared with the
standard one, improves the relative parameters since it decreases the
length but preserves the dimension and the minimum distance.  Though
distinct, Bravyi and Hastings' homological product and (reduced)
tensor products are closely related. Relationship between them are
discussed in Section~\ref{sec:BH}.

We study families of codes obtained by iterated tensor powers of a CSS
code. This operation does not improve the relative parameters but can
reasonably preserve them while providing codes with sparser parity
check matrices.  Actually, for $\qcode{C}_1$ and $\qcode{C}_2$ two CSS
codes, the length and the dimension of the product
$\qcode{C}_1\otimes\qcode{C}_2$ enjoy closed formulae roughly equal to the product of the corresponding parameters of
$\qcode{C}_1$ and $\qcode{C}_2$. The minimum distance
$d_{\qcode{C}_1\otimes\qcode{C}_2}$ of the product is more difficult
to evaluate.
Our main result is a criterion that provides a lower bound, using as
large as possible sets of (co)homologically non trivial elements with
as small as possible {\em overlaps} (see Definition~\ref{def:overlap}):
{
\renewcommand{\thetheo}{\ref{thm:main} and Theorem \ref{theo:mainBH}}
\begin{theo}
Let $\qcode{C}$ be a CSS code defined as a pair of classical
  codes $\code{C}_2\subseteq \code{C}_1^\perp$ given by full rank parity--check matrices.
  Let $g_1, \ldots, g_k\in \code{C}_1^\perp$ and $g_1^*, \ldots, g_k^*\in\code{C}_2^\perp$
  be such that
  $$\code{C}_1^\perp = \code{C}_2 \oplus
  \Span(g_1, \ldots, g_k), \quad 
  \code{C}_2^\perp = \code{C}_1 \oplus
  \Span(g_1^*, \ldots, g_k^*)\quad {\rm and}\quad
  \forall i,j, \langle 
  g_i^*, g_j \rangle = \delta_{ij}.$$
  If, for any $j_0 \in \{1, \ldots , k\}$, there exists $\Omega_{j_0} \subseteq
  g_{j_0}^* + \code{C}_1$ and $\Omega'_{j_0} \subseteq g_{j_0} +
  \code{C}_2$, with $|\Omega_{j_0}|,\ |\Omega'_{j_0}| \geq N$
  and $\overlap (\Omega_{j_0}),\ \overlap (\Omega'_{j_0}) \leq K$. Then,
  for any CSS code $\qcode{D}$:
  $$
  d_{\qcode{C} \otimes \qcode{D}},d_{\qcode{C} \boxtimes \qcode{D}} \geq \left\lceil \frac N K d_{\qcode{D}}
  \right\rceil \cdot
  $$
\end{theo}
\addtocounter{theo}{-1}
}

As a simple application of our criterion, we obtain then
{
\renewcommand{\thetheo}{\ref{cor:minor_times_2} and Corollary \ref{cor:mainCorBH}}
\begin{cor}
  If $\qcode{C}$ and $\qcode{D}$ are two CSS codes described by
  matrices which have no columns of zeros, then
\[
2\max(d_{\qcode{C}},d_{\qcode{D}})\leq
d_{\qcode{C}\otimes\qcode{D}},d_{\qcode{C}\boxtimes\qcode{D}}.
\]
\end{cor}
\addtocounter{theo}{-1}
}
This lower bound is twice better than the 
previously known lower bound \cite[Lemma 2]{Hastings}.
It follows that the iterated
tensor powers of any CSS code described by matrices with no zero
column is an LDPC family whose minimum distances have a non trivial growth tending to infinity:
{
\renewcommand{\thetheo}{\ref{cor:NonTrivial}}
  \begin{cor}
      If $\qcode{C}=\left(\mathbf{H}_X,\mathbf{H}_Z\right)$ is
      any CSS code
  such that none of $\mathbf{H}_X$ or $\mathbf{H}_Z$ has a zero column,
  then the family $\left(\qcode{C}^{\otimes\ell}\right)_{\ell\in\N}$
  is
  logarithmically LDPC with
  $d_{\qcode{C}^{\otimes\ell}}\geq
  2^\ell$ for every $\ell\in\N^*$.
  \end{cor}
\addtocounter{theo}{-1}
}
In particular, the minimum distance grows exponentially fast compared to the row weight of the parity check matrices. So, even if a CSS code has no {\em quantum degeneracy}\footnote{i.e.  its
  quantum minimum distance is not larger than the minimum of the
  distances of the two classical codes defining it.}, for a large
enough $ \ell$, its $\ell$--th iterated power does.

Our criterion for estimating the minimum distance turns out to be
quite efficient when applied with construction involving classical
codes with a large group of automorphism. We give three such examples:
\begin{itemize}
\item binary codes from finite geometry, on which acts
  $\mathbf{PGL}(3,\Fq)$, lead to a CSS code $\QFG(s)$ for every
  $s\in\N^*$;
\item binary cyclic codes of length $n$ on which acts
  $\fract{\Z}/{n\Z}$, lead to a CSS code $\QCC(4^s,2^s)$ for every
  $s\in \N^*$;
\item binary Reed--Muller codes $\RM(r,s)$ on which acts the affine group
  $\mathbf{Aff}(r, \F)$, lead to a CSS code $\QRM(s)$ for each
  $s\in \N^*$.
\end{itemize}
For these three examples, the sequence of iterated $\ell$-th tensor
powers have length $N_\ell$ tending to infinity and minimum distance
which can be larger than $N_\ell^{\alpha}$ for any $\alpha <
\frac{1}{2}$.  Moreover, these codes are {\em logarithmically LDPC},
i.e. they have parity check matrices with row weight in $O(\log
N_\ell)$ and the number of stabilizers acting nontrivially on a qubit
(i.e. the column weight) is in $O(\log N_{\ell})$ too.  The first two
examples provide sequences of CSS codes with constant dimension, while
the third one has a dimension sequence tending to infinity. Moreover,
by diagonal extraction, the latter leads to a family which is almost
LDPC, in the sense that the weights grow slower than $N_\ell^\e$ for
any $\e>0,$ and whose dimensions and minimum distances are
respectively larger than $N_\ell^{\frac{\alpha}2}$ and
$N_\ell^{\alpha}$ for any $\alpha < 1$.

\begin{remarque*}
  One can note that, for all the LDPC families
    provided in this paper, the lower bound for the minimum distance
    culminates at, but does not exceed, the ``square root of the
    length'' barrier. Unfortunately, this is no coincidence, since a
    simple remark (Remark \ref{rk:LaLoose}) shows that if the above
    criterion is sharp for a given code $\qcode C$, then the minimum
    distance of $\qcode C$ \emph{is} at most the square root of the
    length. Without saying anything on the square root barrier
    conjecture in general (even for iterated tensor powers of codes),
    the examples given above are hence somehow optimal as corollaries
    of Theorem \ref{thm:main}.
\end{remarque*}

%  Local Variables: 
%  mode: latex
%  TeX-master: "OnProductCodes"
%  End: 

%% file: 0.2.Organization.tex
Section \ref{sec:Background} contains a brief review of the needed
definitions from homological algebra (Section
\ref{sec:ChainComplexes}), classical codes (Section
\ref{sec:ClassicalCodes}) and CSS codes (Section \ref{sec:CSS}). In
particular, we recall there the deep connection between CSS codes and
chain
complexes.\\
In Section \ref{sec:Main}, we use the latter connection to transport
the notion of tensor product from chain complexes to CSS
codes (Section \ref{sec:CSSpowers}).  We provide then the main
theorem, which gives a lower bound for the minimum distance of the
product of two CSS codes (Section \ref{sec:Th}), and state a number of
direct consequences for rougher, but general, lower bounds
and for parameters of iterated tensor powers (Section \ref{sec:Consequences}).\\
As examples of applications, we provide in Section
\ref{sec:KnownResults} some elementary interpretations, in term of
tensor products, of known results such as the hypergraph product codes
given by J.-P. Tillich and G. Zemor in \cite{TZ} (Section
\ref{sec:TZ}) or Khovanov codes given by the first author in
\cite{Audoux} (Section \ref{sec:Khovanov}). We also relate our tensor
product for CSS codes to the homological product defined by S. Bravyi and
M. Hastings in \cite{Bravyi,Hastings} (Section
\ref{sec:BH}), and we discuss the
product of Steane codes, already discussed in \cite{Bravyi,Hastings} (Section \ref{sec:Steane}).

Note that the relationship of hypergraph and
homological products with tensor products was already noticed in \cite{Freedman}.\\
Finally, section \ref{sec:NewFamilies} is devoted to the description
of three new families of LDPC CSS codes, based on finite geometry
(Section \ref{sec:QFG}), cyclic codes (Section \ref{sec:QCyclic}) and
Reed--Muller codes (Section \ref{sec:QRM}).\\
The paper ends with two technical appendices with the details of the
computation of lengths for iterated tensor powers (Appendix
\ref{appendix:A}) and iterated reduced tensor powers (Appendix
\ref{appendix:B}).

%  Local Variables: 
%  mode: latex
%  TeX-master: "OnProductCodes"
%  End: 

%% file: 0.3.Notation.tex
We shall consider $\F$--spaces, which are finite-dimensional vector
spaces over the field $\F$.  All the theoretical material present in
this paper can actually be adapted to work over any field but, in order to simplify notation, and since
it is sufficient for all the applications we consider here, we
restrict this presentation to the $\F$ case.

For any $\F$--space $C$,
we denote by $C^*:=\Hom(C,\F)$ the dual space of $C$.
Every map $f:A\to B$ induces a dual map
$f^*:B^*\to A^*$ defined by $f^*(\varphi)=\varphi\circ f$
for every $\varphi\in B^*$. For every $X\subseteq C$, we denote its {\em orthogonal space} by $X^\perp:=\{\varphi\in C^*\ |\ \varphi_{|X}\equiv
0\}$.

If $C$ is given with a basis $\BB$, then the bijection
$\big(A\subset\BB\mapsto \psum_{b\in A}b\in C\big)$ identifies the
elements of $C$ with the subsets of $\BB$.  We shall use freely this
identification, denoting subsets $\{a_1,\ldots,a_s\}\subset \BB$, and
the related elements of $C$, by concatenations
$a_1a_2\cdots a_s$.\footnote{note that the order of the $a_i$'s in
  this notation is irrelevant} Associated to $\BB$, there is a natural
dual basis $\BB^*:=\left\{b^*\ |\ b\in\BB\right\}$ for $C^*$, where
$b^*$ is defined by $b^*(b')=\delta_{bb'}$ for all $b'\in\BB$. Here,
$\delta$ stands for the Kronecker delta.
Using the subset
identification mentioned above, we shall denote by $b\in x$, where
$x\in C$ and $b\in \BB$, the fact that $b^*(x)\neq0$, which means that
$b$ appears in the decomposition of $x$. In the same spirit, we denote by $|x|$ the Hamming
weight of $x\in C$, that is the number of $b\in\BB$ such that
$b\in x$.  We shall also denote with brackets the usual bilinear form
defined on $C$ by $\langle b_1,b_2\rangle:=\delta_{b_1b_2}$ for all
$b_1,b_2\in\BB$. The following map:
\[
\fonction{C}{x}{C^*}{y\mapsto\langle x,y\rangle}
\]
is then an isomorphism sending $\BB$ on $\BB^*$.  For every
$X\subset C$, it induces an isomorphism between $X^\perp$ and
$\big\{x\in C\ \big|\ \forall y\in X, \langle x,y\rangle=0\big\}$.
In order to reduce the
amount of notation, we shall use freely this identification without
necessarily mentioning it. The dual of a map $f:A\to B$ would hence be
seen as $f^*:B\to A$.

By convention and unless otherwise specified, $\F$--spaces
shall be denoted using roman capital letters, with an index $i$ when it
corresponds to the degree $i$ part of a graded\footnote{see next section for definitions} space; chain complexes\footnotemark[3] using
cursive capital letters; maps of chain complexes by $\p$, possibly with a distinctive
index or exponent; quantum codes\footnotemark[3] using calligraphic capital
letters; and classical codes\footnotemark[3] using calligraphic
capital letters of a slightly modified type.
A same letter shall be used for associated objects: typically
$\qcode{C}$ shall be the CSS code\footnotemark[3] associated to the chain complex
${\chain{C}}$ defined as the $2$--nilpotent\footnotemark[3] map $\p$
(or $\p_{\chain{C}}$) defined on
$C:=\poplus_{i\in\Z}C_i$. The map $\p_i$ shall be then the restricted
map $\p_{|C_i}$. If a classical code is involved in the story, then it
should be $\code{C}$.

%  Local Variables: 
%  mode: latex
%  TeX-master: "OnProductCodes"
%  End: 

%% file: 0.4.Acknowledgement.tex
The present paper is the fruit of several meetings organized by the
TOCQ team and funded by the CNRS (for the PEPS ICQ 2013 grant TOCQ), the
Labex Archimède and the INRIA.
The authors are hence willing to thank the other members of
the team for all the dicussions and all their comments on the paper:
Anthony Leverrier, Jean-Pierre Tillich and Gilles Zémor and especially
Nicolas Delfosse for his relevant observations.
The authors also express their gratitude to Matthew Hastings for his
feedback on a first version of the present article and for interesting
discussions. Finally, they thank the anonymous referre for all its comments.

%  Local Variables: 
%  mode: latex
%  TeX-master: "OnProductCodes"
%  End: 

%% file: 1.1.ChainComplexes.tex
\subsubsection{Definitions}

For the sake of self-containedness, we begin by a review of standard
notions of homological algebra (see {\it e.g.} \cite{Weibel} for
further details).

In the literature, chain complexes are often defined as a sequence of
$\F$--spaces ${(C_i)}_{i \in \Z}$ which are all zero but a finite
number of them, together with a collection of maps either all of the
form $\p_i : C_i \rightarrow C_{i+1}$ or all of the form
$\p_{i} : C_i \rightarrow C_{i-1}$.  Another way to describe them is
to consider the direct sum $C :=\poplus_{i \in \Z} C_i$ and regard the
collection of maps ${(\p_i)}_{i \in \Z}$ as a graded endomorphism of
$C$. In the present paper, we shall adopt the latter approach.

\begin{defi}\label{def:ChainComplex}
  A linear map $\p\in\End(C)$, for some $\F$--space $C$, is $2$--nilpotent
  if it satisfies $\p^2=0$.\\
  An $\e$--chain complex ${\chain{C}}$, for $\e=\pm1$, is a $2$--nilpotent
  map $\p\in\End(C)$
  such that
  \begin{itemize}
  \item $C$ is $\Z$--graded, that is decomposes into
    $C:=\poplus_{i\in\Z}C_i$;
  \item $\p$ increases the degree by exactly
    $\e$, that is $\Im(\p_{|C_i})\subset C_{i+\e}$ for every $i\in\Z$.
  \end{itemize}

  If ommited and unless otherwise specified, $\e$ shall be assumed to
  be equal to $1$.\\
  Since $C$ is finite-dimensional,
  there is only a finite
  number of degrees $i$ such that $C_i \neq \{0\}$.
  The {\em support} of a chain
  complex is the smallest interval $\{a,a+1,\ldots,b\}$ of integers such that
  $C_i = \{0\}$ for
  $i<a$ or $i>b$ and the value $b-a+1$ is called the {\em length} of the
  chain complex.\\
A basis $\BB$ for ${\chain{C}}$ is the data of a basis for
  each non zero space $C_i$, that is an identification of $C_i$ with a
  power of $\F$.
\end{defi}

\begin{nota}\label{not:Complex}
  Chain complexes shall be represented as
\[
\xymatrix{
\cdots \ar[r]^-{\p_{i-2}} & C_{i-1} \ar[r]^-{\p_{i-1}} & C_{i} \ar[r]^-{\p_{i}} &
C_{i+1} \ar[r]^-{\p_{i+1}} & C_{i+2} \ar[r]^-{\p_{i+2}} & \cdots 
}.
\]
In explicit cases given with a basis, $C_i$ shall be represented by
dots, one for each generator, and $\p_i$ shall be represented by edges
joining a generator $x$ to the elements of $\p_i(x)$. For instance,
the following picture:
\[
\vcenter{\hbox{$\xymatrix@!0 @R=1cm @C=2.5cm{
&\bullet\ar@{-}[r]\ar@{-}[dr]&\bullet\ar@{-}[dr]&\\
\bullet\ar@{-}[ur]\ar@{-}[r]\ar@{-}[dr]&\bullet\ar@{-}[ur]\ar@{-}[dr]&\bullet\ar@{-}[r]&\bullet\\
&\bullet\ar@{-}[ur]\ar@{-}[r]&\bullet\ar@{-}[ur]&
}$}}
\]
represents the complex $\xymatrix{\Span(w_1) \ar[r]^-{\p_0} &
  \Span(x_1,x_2,x_3) \ar[r]^-{\p_1} & \Span(y_1,y_2,y_3)
  \ar[r]^-{\p_2} & \Span(z_1)}$, where
$\Span$ denotes the vector space spanned by the given generators
and where
$\p_0 (w_1) =
x_1+x_2+x_3$, $\p_1 (x_1,x_2,x_3) = (y_1+y_2, y_1+y_3, y_2+y_3)$ and
$\p_2(y_1)=\p_2(y_2)=\p(y_3) =z_1$.
\end{nota}
\begin{defi}\label{def:Dual}
  For any $\e$--chain complex ${\chain{C}}$, we define its dual ${\chain{C}}^*$
  as the $(-\e)$--chain complex $\p^*\in\End(C^*)$ defined by
  $C^*:=\poplus_{i\in\Z}\Hom(C_i,\F)$ and $\p^*(\varphi)=\varphi\circ\p$
  for every $\varphi\in C^*$.

We say that a chain complex is \emph{symmetric} if it is isomorphic,
as a chain complex, to its dual.
\end{defi}
\begin{prop}\label{prop:DualTranspose}
  If $\BB$ is a basis for an $\e$--chain complex
  ${\chain{C}}$, then $\Mat_{\BB^*}(\p^*)={}^t\Mat_\BB(\p)$, where
  $\Mat_\BB(f)$ denotes the matrix representing the
  linear map in the basis $\BB$, with the convention that columns are
  the images of the generators, and ${}^t\Mat_\BB(\,.\,)$ denotes its
  transpose.
\end{prop}

\begin{remarque}
  If $\chain{C}$ is given with a basis, then the maps $\p_i$ can be given
  by their matrices. Using the identification
  between an $\F$--space and its dual mentioned in the Notation
  section, $\chain{C}^*$ can be seen as the chain complex obtained by
  reversing all the arrows and transposing all the matrices.\\
  Furthermore, over $\F$, Proposition \ref{prop:DualTranspose} is 
  proven by noting that, for every pair of generators $x$ and $y$,
\[
y\in\p(x)\ \Longleftrightarrow\ 
y^*\big(\p(x)\big)\neq0\ \Longleftrightarrow\ \p^*(y^*)(x)\neq0\
\Longleftrightarrow\ 
x^*\in\p^*(y^*).
\]
So, if $\chain{C}$ is given using Notation \ref{not:Complex}, then $\chain{C}^*$ is obtained by reading
the graph from right to left.
\end{remarque}

\begin{defi}
  For any $\e$--chain complex ${\chain{C}}$ and any integer $i\in\Z$, we define its
 $i^\textrm{th}$  homology group as
 $H_i({\chain{C}}):=\fract{\Ker(\p_i)}/{\Im(\p_{i-\e})}$, and set $H_\bullet({\chain{C}}):=\poplus_{i\in\Z}H_i({\chain{C}})$.
  For any $x\in\Ker(\p)$, we denote by $[x]$ its image in $H_\bullet({\chain{C}})$.
\end{defi}

\begin{defi}
  If ${\chain{C}}$ is an $\e$--chain complex given with a basis $\BB$, then, for each $i\in\Z$ we denote by
  \begin{itemize}
  \item $n_i({\chain{C}}):=\dim(C_i)$ and define the \emph{length} of
    $\chain{C}$ as $n_\chain{C}:=n_0(\chain{C})$;
  \item $k_i({\chain{C}}):=\dim\big(H_i({\chain{C}})\big)$ and define the \emph{dimension} of
    $\chain{C}$ as $k_\chain{C}:=k_0(\chain{C})$;
  \item $d_i({\chain{C}}):=\min\big\{|x|\ \big|\ [x]\in
    H_i({\chain{C}})\setminus\{0\}\big\}$ and define the \emph{minimum distance} of
    $\chain{C}$ as $d_\chain{C}:=d_0(\chain{C})$;
  \item $w_i({\chain{C}}):=\max\big\{|x|\ \big|\ x \textrm{ row of }\Mat_\BB(\p_i)\big\}$ and define the \emph{weight} of
    $\chain{C}$ as $w_\chain{C}:=w_0(\chain{C})$.
  \end{itemize}
\end{defi}
\begin{remarque}\label{rem:Dependances}
   The above parameters have only a relative dependency with regard to
  the basis. Indeed, $w_{\chain{C}}$ depends on the entire $\BB$, $d_{\chain{C}}$
  depends only on its restriction $\BB_{|C_{0}}$, $n_{\chain{C}}$ and $k_{\chain{C}}$ are independent of $\BB$. \end{remarque}

\subsubsection{Operations on chain complexes}

\begin{defi}
Let ${\chain{C}}$ and ${\chain{D}}$ be two
$\e$--chain complexes. We define  their direct sum ${\chain{C}}\oplus{\chain{D}}$ as the $\e$--chain
complex $\p_{\chain{C}}\oplus\p_{\chain{D}}\in\End\Big(\poplus_{i\in\Z}\big(C_i\oplus
D_i\big)\Big)$.
\end{defi}

\begin{prop}
    Let ${\chain{C}}$ and ${\chain{D}}$ be two $\e$--chain complexes
    given with basis. Then
  \begin{itemize}
  \item $\big({\chain{C}}\oplus{\chain{D}}\big)^*\cong{\chain{C}}^*\oplus{\chain{D}}^*$;
  \end{itemize}
and for each $i\in\Z$,
  \begin{itemize}
  \item $H_i({\chain{C}}\oplus{\chain{D}})\cong H_i({\chain{C}})\oplus H_i({\chain{D}})$;
  \item $n_i({\chain{C}}\oplus{\chain{D}})=n_i({\chain{C}})+n_i({\chain{D}})$,
    $k_i({\chain{C}}\oplus{\chain{D}})=k_i({\chain{C}})+k_i({\chain{D}})$,
    $d_i({\chain{C}}\oplus{\chain{D}})=\min\big(d_i({\chain{C}}),d_i({\chain{D}})\big)$
    and
  \item $w_i({\chain{C}}\oplus{\chain{D}})=\max\big(w_i({\chain{C}}),w_i({\chain{D}})\big)$. 
  \end{itemize}
\end{prop}

\begin{proof}
All the statements, except the one on minimum distances and the one on
weights, are classical results of homological algebra.

For the statement on minimum distances, let $x \in C_i$ and $y \in D_i$ be homologically non trivial
elements of minimum weight. Then $(x,0)$ and $(0,y) \in C_i \oplus D_i$
are homologically non trivial, so
$d_i({\chain{C}}\oplus{\chain{D}})\leq\min\big(d_i({\chain{C}}),d_i({\chain{D}})\big)$.
Conversely, every homologically non trivial
 element $(a,b) \in C_i \oplus D_i$ is such that either $a$ or $b$
is homologically non trivial and its weight is hence larger than either
the weigth of
$(x,0)$ or of $(0, y)$.

For the statement on weights, consider the maps $\partial_{\chain{C}, i} : C_i
\rightarrow C_{i+1}$ and $\partial_{\chain{D}, i} : D_i \rightarrow D_{i+1}$.
Let $M_{\chain{C}, i}$ and $M_{\chain{D}, i}$ be their matrix representations.
Then, the map $\partial_{\chain{C} \oplus \chain{D}, i} : C_i \oplus D_i \rightarrow C_{i+1}\oplus D_{i+1}$ is represented by the matrix:
\[
\begin{pmatrix}
  M_{\chain{C}, i} & 0 \\
  0              & M_{\chain{D}, i}
\end{pmatrix}
\]
which yields the result on the maximum weight of the rows.
\end{proof}

In particular, we emphasize the fact that adding a direct
summand given with its own basis and which has null homology does not
affect the parameters except the length which is increased
consequently.
But conversely, detecting and removing a direct summand may alter the
minimum distance if the basis does not respect the direct sum decomposition.

\begin{defi}
Let ${\chain{C}}$ and ${\chain{D}}$ be two
$\e$--chain complexes. We define the tensor product ${\chain{C}}\otimes{\chain{D}}$ as the $\e$--chain
complex
$\Id_C\otimes\p_{\chain{D}}+\p_{\chain{C}}\otimes\Id_D\in\End\Big(\poplus_{i\in\Z}\big(\poplus_{r\in\Z}(C_r\otimes
D_{i-r})\big)\Big)$.
\end{defi}
\begin{prop}\label{prop:ProductParameters}
    Let ${\chain{C}}$ and ${\chain{D}}$ be two $\e$--chain complexes. Then
  \begin{itemize}
  \item $\big({\chain{C}}\otimes{\chain{D}}\big)^*\cong{\chain{C}}^*\otimes{\chain{D}}^*$;
  \end{itemize}
and for each $i\in\Z$,
  \begin{itemize}
  \item $H_i({\chain{C}}\otimes{\chain{D}})\cong\poplus_{r\in\Z}\big(H_r({\chain{C}})\otimes
    H_{i-r}({\chain{D}})\big)$ (K\"unneth formula);
  \item $n_i({\chain{C}}\otimes{\chain{D}})=\disp{\sum_{r\in\Z}}n_r({\chain{C}}).n_{i-r}({\chain{D}})$,
    $k_i({\chain{C}}\otimes{\chain{D}})=\disp{\sum_{r\in\Z}}k_r({\chain{C}}).k_{i-r}({\chain{D}})$;
  \item if ${\chain{C}}$ and ${\chain{D}}$ were given with bases
    $\BB_\chain{C}$ and $\BB_\chain{D}$, then
    $\BB_\chain{C}\otimes\BB_\chain{D}$ provides a basis for
    $\chain{C}\otimes\chain{D}$ such that
    $w_i({\chain{C}}\otimes{\chain{D}})=\max\big\{w_j(\chain{C}) +
    w_k (\chain{D})\big| j+k = i\big\}$.
  \end{itemize}
\end{prop}

\begin{proof}
The statements on duals, homologies and lengths are classical results of homological algebra.
The statement on weights follows from
$\Mat_{\BB_\chain{C}\otimes\BB_\chain{D}}(\p_{\chain{C}\otimes\chain{D}})$
being obtained as a sum of Kronecker products of the form
$\Mat_{\BB_\chain{C}}(\p_{\chain{C}}) \otimes \textrm{Id}$ and
$\textrm{Id} \otimes \Mat_{\BB_\chain{D}}(\p_{\chain{D}})$. 
\end{proof}

The isomorphism of K\"unneth formula, proven for instance in \cite[Thm
3.6.3]{Weibel}\footnote{Noting that, since we are working over a
  field, any module is flat and hence 
  Tor is zero.}, is actually induced from maps defined at the chain
complex level. It induces hence the following proposition which shall
be needed further in the proof of Lemma \ref{lem:Th}. 

\begin{prop}\label{prop:TechKunneth}
  If, for each $r\in\Z$, $x^r_1,\ldots x^r_{j_r}\in C_r$ and
  $y^r_1,\ldots y^r_{j'_r}\in D_r$ induce a basis for, respectively,
  $H_r(\chain{C})$ and $H_r(\chain{D})$, then, for every $i\in\Z$, the elements of the form
  $x_j^{r}\otimes y_{j'}^{i-r}$ induce a basis for $H_i(\chain{C}\otimes\chain{D})$. 
\end{prop}

\begin{remarque}
  Evaluating $d_i({\chain{C}}\otimes{\chain{D}})$ is less
  straightforward and it shall be the aim of Section \ref{sec:Th}.
\end{remarque}

\subsubsection{Short complexes and reduction}
\label{sec:Reduction}

In this paper, we shall be mostly interested in length 3 chain
complexes centered around degree zero. This motivates the following definitions.
\begin{defi}\label{def:ShortEtc}
  A chain complex ${\chain{C}}$ is said to be a \emph{short complex}
  if it has a support contained in $\{-1,0,1\}$. It is hence of the
  following form:
  $\xymatrix{C_{-1}\ar[r]^-{\p_{-1}}&C_{0}\ar[r]^-{\p_{0}}&C_{1}}$.
The chain complex is said to be \emph{balanced}
if it has non trivial
homology only in degree zero.
A balanced short complex is said to be \emph{reduced}. For a short
complex, being reduced is equivalent to require
$\p_{-1}$ to be injective and $\p_{0}$ to be surjective. So the chain complex is of the
form $\xymatrix{C_{-1}\ar@{^(->}[r] & C_{0} \ar@{->>}[r] &
  C_{1}}$. Equivalently, it consists in requiring that
$\dim\big(H_{0}({\chain{C}})\big)=\dim(C_{0})-\dim(C_{-1})-\dim(C_{1})$.

Note that a short complex $\xymatrix{
  C_{-1} \ar[r]^{\p_{-1}} & C_0 \ar[r]^{\p_0} & C_1}$ is symmetric if
and only if $C_{-1}\simeq 
C_1^*$ and $\p_0 = \p_{-1}^*$.
\end{defi}

Any chain complex $\chain{C}$ can be turned into a short one by truncating the
degrees higher than 1 and lower than $-1$. More precisely, by shifting
beforehand the
degree, one can extract any length
3 portion of $\chain{C}$.
However, the result is generally not balanced, even if $\chain{C}$ was.
There is nonetheless a reduction process to turn a short complex into
a reduced one (almost) without altering its parameters.
Indeed, if
  ${\chain{C}}:=\xymatrix{C_{-1}\ar[r]^{\p_{-1}}&C_{0}\ar[r]^(.45){\p_{0}}&C_{1}}$
  is given with a basis $\BB$, and if ${\chain{D}}$
  is obtained from ${\chain{C}}$ by removing all redundant rows
 of ${}^t\Mat_\BB(\p_{-1})$ and/or $\Mat_\BB(\p_{0})$ and by modifying
 $C_{-1}$ and $C_{1}$ consequently, then $\chain{D}$ is reduced and it is mostly a consequence of
 Remark \ref{rem:Dependances} that $n_{\chain{D}}=n_{\chain{C}}$, $k_{\chain{D}}=k_{\chain{C}}$,
  $d_{\chain{D}}=d_{\chain{C}}$ and $w_{\chain{D}}\leq
  w_{\chain{C}}$. From a linear algebraic point of view, it consists in replacing $C_{-1}$
  by a complement space for $\Ker(\p_{-1})$ spanned by vectors of $\BB$, and replacing $C_1$ by
  its quotient under a complement space of $\Im(\p_0)$ spanned by vectors of $\BB$.
  This process is however non canonical
  since it requires the choice of complement spaces, or equivalently, the
  choice of the redundant rows to be removed.

As a conclusion, any length 3 portion of a chain complex can be
grading-shifted so it is centred in
degree zero, and then the above (non canonical) reduction process can turn
it into a reduced complex without altering the
parameters, except the weight which may even be decreased.

\begin{remarque}\label{rem:AltReduc}
  Adding $\Ker(\p_{-1})$ in degree $-2$ and $\Coker(\p_0)$ in degree
  2 is another way to turn a short complex into a balanced one. The
  chain complex is then of length five. This provides a
  balancing 
  process which is canonical
  and which preserves all the parameters. However, for length reasons, we shall consider in this paper, only
  the (non canonical) reduction process, and not the (canonical)
  balancing one.
\end{remarque}

%  Local Variables: 
%  mode: latex
%  TeX-master: "OnProductCodes"
%  End: 

%% file: 1.2.ClassicalCodes.tex
As they shall play a keyrole in several constructions, we set here some
notation on classical codes.
A classical code $\code{C}$ is a subspace of an
$\F$--space $E$ given with a
basis $\BB_E$. It can be described by either a generating map
$g_{\code{C}}:\xymatrix{A\ar@{^(->}[r]&E}$ such that
$\Im(g_{\code{C}})=\code{C}$ or a parity-check map
$p_{\code{C}}:\xymatrix{E\ar@{->>}[r]&B}$ such that
$\Ker(p_{\code{C}})=\code C$. For any such code, we define:
\begin{itemize}
\item its length $n_{\code{C}}$ as the dimension of $E$;
\item its dimension $k_{\code{C}}$ as the dimension of $\code C$;
\item its minimum distance $d_{\code{C}}$ as the minimum weight for a
  non trivial element of $\code C$, using the basis $\BB_E$;
\item its weight as the maximal weight of a row of
  $\Mat_{\BB_E,\BB_B}(p_{\code{C}})$, where $\BB_{B}$ is a given basis
  for $B$. 
\end{itemize}

We define the {\em dual} of $\code{C}$ as the code $\code{C}^\perp$
defined by $\code{C}^\perp\subset E^*$, which can also be seen,
using the identification mentioned in the Notation section, as
$\big\{x\in E\ \big|\ \forall y\in
\code C, \langle x,y\rangle=0\big\}\subset E$.
It is easily checked that $p^*_{\code{C}}$ and $g^*_{\code{C}}$ are,
respectively, a generating map and a parity-check map for
$\code{C}^\perp$ so that, up to transpose, $\code{C}$ and $\code{C}^\perp$
exchange their generating and parity-check matrices; and that $n_{\code{C}^\perp}=n_{\code{C}}$ and $k_{\code{C}^\perp}=n_{\code{C}}-k_{\code{C}}$.

%  Local Variables: 
%  mode: latex
%  TeX-master: "OnProductCodes"
%  End: 

%% file: 1.3.CSSCodes.tex
CSS codes were developped in \cite{Calderbank,Steane}. They are a special case of
stabilizer quantum error
correcting codes associated to pairs of orthogonal classical
codes, that is codes $\code{C}_1$ and $\code{C}_2$ such that
$\code{C}_2\subseteq\code{C}^\perp_1$; or equivalently to matrices $\mathbf{H}_X$ and $\mathbf{H}_Z$ such that $\mathbf{H}_X {}^t\mathbf{H}_Z=0$. A quick review can be found in section 1.1 of \cite{Audoux} but
for a more comprehensive treatment, we refer the reader to
\cite{Nielsen,Preskill,Delfosse}.
  A CSS code is said to be {\em symmetric} if $\code{C}_1 = \code{C}_2$ or, equivalently,
  if $\mathbf{H}_X = \mathbf{H}_Z$.

  \begin{remarque}
    The terminology of {\em symmetric CSS} codes is non standard. Such
    codes are sometimes referred to as {\em weakly self dual CSS
      codes} in the literature.  We preferred use the term {\em
      symmetric} since it is coherent with our terminology of
    symmetric chain complexes. Indeed, if a chain complex is
    symmetric, then the corresponding CSS code is symmetric.
  \end{remarque}
  
For a CSS code $\qcode{C}$, some relevant parameters are
\begin{itemize}
\item $n_{\qcode{C}}$ the length of ${\qcode{C}}$,
that is the common length of
the codes $\code{C}_1$ and $\code{C}_2$ ;
\item $k_{\qcode{C}}$ the dimension of ${\qcode{C}}$, that is the dimension of $\fract{\code{C}_1^{\bot}}/{\code{C}_2}$
\item $d_{\qcode{C}}$ the minimum distance of ${\qcode{C}}$, that is
  the minimum weight of an element of $(\code{C}_1^\perp\setminus\code{C}_2)\cup (\code{C}_2^\perp\setminus\code{C}_1)$;
\item $w_{\qcode{C}}$ the weight of ${\qcode{C}}$, that is the highest weight realized by a row of $\mathbf{H}_X$ or $\mathbf{H}_Z$.
\end{itemize}
They shall be gathered in the notation $\llbracket
n_{\qcode{C}};k_{\qcode{C}};d_{\qcode{C}};w_{\qcode{C}}\rrbracket$.

Chain complexes turn out to be efficient for constructing such CSS
codes. Indeed, once equipped with a basis, they not only
naturally provide matrices whose product is zero, but
parameters can also be read from them, their duals and the associated homologies.
The following classical statement reformulates the usual matrix-based
description of CSS codes in terms of chain complexes, and this allows a more intrinsic
description of these objects. Similar statements appear,
for instance, in \cite[Prop. 1.7]{Audoux} or \cite{Delfosse}.

\begin{prop}\label{prop:ChainToCSS}
  To a short complex
  ${\chain{C}}:=\xymatrix{C_{-1}\ar[r]^{\p_{-1}}&C_{0}\ar[r]^(.45){\p_{0}}&C_{1}}$
  given with a basis $\BB$, there is an associated CSS code
  ${\qcode{C}}:=\left(\Mat_\BB(\p_{0}),{}^t\Mat_\BB(\p_{-1})\right)$
  with parameters $n_{\qcode{C}}=n_{\chain{C}}$,
  $k_{\qcode{C}}=k_{\chain{C}}$,
  $d_{\qcode{C}}=\min\big(d_{\chain{C}},d_{{\chain{C}}^*}\big)$ and
  $w_{\qcode{C}}=\max\big(w_{\chain{C}},w_{{\chain{C}}^*}\big)$.
\end{prop}

Compared to the definition of CSS codes given above, the
codes $\code{C}_1$ and $\code{C}_2$ correspond to $\code C_1  = \Ker (\p_{0})^\perp$ and
$\code C_2 = \Im (\p_{-1})$, or equivalently to $\code C_1 = \Im(\p_{0}^*)$
and $\code C_2 = \Ker (\p^*_{-1})^\perp$.
Conversely, two matrices $\mathbf{H}_X$ and $\mathbf{H}_Z$ such that
$\mathbf{H}_X {}^t\mathbf{H}_Z=0$ provide a chain complex
\[
\xymatrix{0\ar[r]&\F^{k_1}\ar[r]^-{{}^t\mathbf{H}_Z}&\F^{n_{\qcode{C}}}\ar[r]^-{\mathbf{H}_X}&\F^{k_2}\ar[r]&0}
\]
\noi where $k_1$ and $k_2$ are, respectively, the numbers of rows in
$\mathbf{H}_Z$ and $\mathbf{H}_X$. There is hence a one-to-one
correspondence between CSS codes and, up to isomorphisms, short complexes given with a
basis. However, as a consequence of the
discussion on reduction given in Section \ref{sec:Reduction}, removing redundant rows in
$\mathbf{H}_Z$ and $\mathbf{H}_X$ does not affect the parameters,
except the weight which may even be decreased. It is hence natural to focus on CSS codes associated to reduced complexes.

\begin{remarque}\label{rem:BH}
The data of a $2$--nilpotent map $\p\in\End(C)$
is actually sufficient to construct a CSS code as $\code{C}_1 =
\Ker(\p)^\perp$ and $\code{C}_2 = \Im(\p)$. This code is actually the code 
associated to the
short complex $\xymatrix{C \ar[r]^{\p} & C\ar[r]^{\p} & C}$.
Note that from every pair of classical codes $\code{C}_1,
\code{C}_2$ such that $\code{C}_2 \subseteq \code{C}_1^{\bot}$,
 one can always construct a 2--nilpotent map $\p : \F^n
\rightarrow \F^n$ whose image is $\code{C}_2$ and kernel is
$\code{C}_1^{\bot}$.
This means that every quantum code can be represented by a $2$--nilpotent
chain complex.
This description of quantum CSS code from $2$--nilpotent
map is used in \cite{Hastings} and shall be discussed in Section \ref{sec:BH}.
\end{remarque}

%  Local Variables: 
%  mode: latex
%  TeX-master: "OnProductCodes"
%  End: 

%% file: 2.1.CSSProduct.tex
As mentioned in the previous section, CSS codes are in one-to-one correspondence
with short complexes. As such, they inherit the notions of
direct sum and tensor product. The former is well defined since it
sends short (respectively reduced) complexes to short (respectively reduced)
complexes. The latter requires some more attention.
Indeed, the tensor product of two short complexes is, in general, not
short anymore but of length 5. One way to correct this shortcoming is
to roughly truncate.
\begin{defi}
  For $\big(\qcode{C}_i\big)_{i\in I}$ a finite family of CSS codes, we define
  $\potimes_{i\in I}\qcode{C}_i$ as the CSS code associated to the
  degrees $\{-1,0,1\}$--truncation of $\potimes_{i\in I}\chain{C}_i$,
  where, for each $i\in I$, $\chain{C}_i$ is the short complex
  associated to $\qcode{C}_i$. When the family is made of $\ell$
  copies of the same code
  $\qcode{C}$, we also denote it by $\qcode{C}^{\otimes\ell}$.
\end{defi}

However, the truncation produces two major drawbacks:
\begin{itemize}
\item the operation is not associative in the sense that, in general,
  $(\potimes_{i\in I}\qcode{C}_i)\otimes (\potimes_{i\in
    J}\qcode{C}_i)\ncong (\potimes_{i\in
    I\sqcup J}\qcode{C}_i)$;
\item the result is, in general, not reduced even if all the factors are. 
\end{itemize}

To remedy the second issue, one can use the reduction process described in
Section \ref{sec:Reduction}.

\begin{defi}
  For $\big(\qcode{C}_i\big)_{i\in \Inter{1}{k}}$ a finite family
  of CSS codes, we define recursively
  $\protimes_{i\in \Inter{1}{k}}\qcode{C}_i$ as the CSS code associated to the
  reduction of the
  degrees $\{-1,0,1\}$--truncation of $\chain{C}_{1\cdots (k-1)}\otimes \chain{C}_{k}$,
  where $\chain{C}_{1\cdots(k-1)}$ is the reduced complex
  associated to $\protimes_{i\in \Inter{1}{k-1}}\qcode{C}_i$ and $\chain{C}_{k}$ the short complex
  associated to $\qcode{C}_{k}$. When the family is made of $\ell$
  copies of the same code
  $\qcode{C}$, we also denote it by $\qcode{C}^{\rotimes\ell}$.
\end{defi}

\begin{warning}
  The notation $\rotimes$ is an abuse of notation since it is not
  canonically defined and requires, at each step, the choice of
  redundant rows to be removed.
\end{warning}

\begin{remarque}
  In order to get a canonical notion of somehow reduced tensor product, one
  can use the balancing process mentioned in Remark
  \ref{rem:AltReduc}. However, it ends with slightly longer codes.
\end{remarque}

%  Local Variables: 
%  mode: latex
%  TeX-master: "OnProductCodes"
%  End: 

%% file: 2.2.Lemma.tex
Let ${\chain{C}}$ be a chain complex given with bases
$(a_1,\ldots,a_{n_{-1}})$ for $C_{-1}$, $(b_1,\ldots,b_{n_0})$ for
$C_0$ and $(c_1,\ldots,c_{n})$ for
$\poplus_{i\in\Z\setminus\{-1,0\}} C_i$.  We denote the matrix
associated to $\p_{-1}$ by
\[
  \left(\begin{array}{ccc}\lambda_{11}&\cdots&\lambda_{1n_{-1}}\\\vdots&&\vdots\\\lambda_{n_01}&\cdots&\lambda_{n_0n_{-1}}\end{array}\right)=\Bigg(\Lambda_1,\ldots,\Lambda_{n_{-1}}\Bigg).
\]
We fix $g_{1},\ldots,g_{r}$ elements in $\Ker(\p_0)$ which
generate a basis of $H_0({\chain{C}})$, and set
$g_{j}:=\disp{\sum_{i=1}^{n_0}}\gamma_i^jb_i$ for
all $j\in\Inter{1}{r}$. Then we complete $g_{1},\ldots,g_{r}$ in two steps, first into a
basis of $\Ker(\p_0)$, and then into a basis of $C_0$.
Finally we define, for all $j_0\in\Inter{1}{r}$,
\[
\Ker_{j_0}^\cperp:=g^*_{j_0}+\Ker(\p_0)^\perp
\]
which, roughly speaking, corresponds to the
elements in $C_0$ which are orthogonal to any generator of
$\Ker(\p_0)$ but $g_{j_0}$.\footnote{Note that, since we are working
  over $\F$, $\Ker_{j_0}^\cperp$ is an affine subspace, but over any
  other field, it should be the union of affine subspaces
  $\pcup_{\lambda\in\mathds{K}^*}\lambda g^*_{j_0}+\Ker(\p_0)^\perp$.}

  \begin{remarque}\label{rk:CohomNonTriv}
    Elements of $\Ker_{j_0}^\cperp$ are in
    $\Im(\p_{-1})^\perp=\Ker(\p_{-1}^*)$, but not in
    $\Ker(\p_0)^\perp=\Im(\p_0^*)$, they are hence cohomologically non
    trivial elements.
  \end{remarque}
  
\begin{defi}\label{def:overlap}
For every subset $\Omega\subset C_0$, we define $\overlap(\Omega):=\displaystyle{\max_{i\in\Inter{1}{n_0}}}\Big|\big\{p\in\Omega\ \big|\
  b_i\in p\big\}\Big|$. If stacking, as rows of a matrix, the elements of
$\Omega$, it corresponds to the maximum weight of a column.
\end{defi}

\begin{lemme}\label{lem:Th}
  Let $N,K\in\N^*$.
  If, for any $j_0\in\Inter{1}{r}$, there exists $\Omega_{j_0}\subset
  \Ker^{\cperp}_{j_0}$ such that $|\Omega_{j_0}|\geq N$ and
  $\overlap(\Omega_{j_0})\leq K$,
  then, for every chain complex ${\chain{D}}$ such that either $\chain{C}$
  or $\chain{D}$ is balanced,
\[
d_{{\chain{C}}\otimes{\chain{D}}}\geq \left\lceil\frac{N}{K}d_{\chain{D}}\right\rceil.
\]
\end{lemme}

\begin{proof}
  We consider
  $x_0=\disp{\sum_{i=1}^{n_{-1}}a_i\otimes\mathfrak{a}_i}+\disp{\sum_{i=1}^{n_0}b_i\otimes\mathfrak{b}_i}+\disp{\sum_{i=1}^{n}c_i\otimes\mathfrak{c}_i}$,
  with $\mathfrak{a}_i,\mathfrak{b}_i,\mathfrak{c}_i\in{\chain{D}}$, a
  minimally weighted representative of a non trivial class in
  $H_0({\chain{C}}\otimes{\chain{D}})$. As a consequence of K\"unneth
  formula, and since either $\chain{C}$ or $\chain{D}$ is balanced,
  $H_0({\chain{C}}\otimes{\chain{D}})\cong H_0(\chain{C})\otimes
  H_0(\chain{D})$ and there exist
  $\mathfrak{g}_1,\ldots,\mathfrak{g}_r\in\Ker(\p_{\chain{D}})$ and
  elements
  $\mathfrak{x}_j,\mathfrak{y}_i,\mathfrak{z}_j\in{\chain{D}}$, such
  that $[\mathfrak{g}_{j_0}]$ is non trivial in $H_0(\chain{D})$ for
  at least one given $j_0\in\Inter{1}{r}$ and

\begin{eqnarray*}
    x_0&=&\sum_{j=1}^rg_{j}\otimes
           \mathfrak{g}_{j}+\p_{{\chain{C}}\otimes{\chain{D}}}\left(\sum_{j=1}^{n_{-1}}a_j\otimes
           \mathfrak{x}_j+\sum_{i=1}^{n_0}b_i\otimes
           \mathfrak{y}_i+\sum_{i=1}^{n}c_i\otimes \mathfrak{z}_i\right)\\
&=&\sum_{j=1}^r\sum_{i=1}^{n_0}\gamma_i^jb_i\otimes \mathfrak{g}_{j}+\sum_{j=1}^{n_{-1}}\sum_{i=1}^{n_0}\lambda_{ij}b_i\otimes \mathfrak{x}_j+\sum_{i=1}^{n_0}b_i\otimes\p_{{\chain{D}}}(\mathfrak{y}_i)+\textrm{ non relevant terms},
  \end{eqnarray*}
where the non relevant terms are of the form $a_i\otimes \cdot\ $ or
$c_i\otimes \cdot\ $ for some $i$.

For every $i\in\Inter{1}{n_0}$, we project on terms of the form
$b_i\otimes \cdot\ $ and obtain then
\[
\mathfrak{b}_i=\disp{\sum_{j=1}^r}\gamma_i^j \mathfrak{g}_{\!j}+\disp{\sum_{j=1}^{n_{-1}}\lambda_{ij}\mathfrak{x}_j}+\p_{\chain{D}}(\mathfrak{y}_i).
\]

For any $p=(p_1,\ldots,p_{n_0})\in\F^{n_0}$, we have then 
\[
\sum_{i=1}^{n_0}p_i \mathfrak{b}_i=\sum_{j=1}^r\sum_{i=1}^{n_0}p_i\gamma_i^j \mathfrak{g}_{j}+\sum_{j=1}^{n_{-1}}\sum_{i=1}^{n_0}p_i\lambda_{ij}\mathfrak{x}_j +\p_{\chain{D}}(\textrm{something}),
\]
which can be reformulated as
\[
\sum_{i\in p}\mathfrak{b}_i=\sum_{j=1}^r\left\langle
  p,g_{j}\right\rangle \mathfrak{g}_{j}+\sum_{j=1}^{n_{-1}}\langle
p,\Lambda_j\rangle \mathfrak{x}_j+\p_{\chain{D}}(\textrm{something}).
\]
If $p\in\Ker^{\cperp}_{j_0}$, we obtain $\disp{\sum_{i\in p}}\ \mathfrak{b}_i=\mathfrak{g}_{j_0}+\p_{\chain{D}}(\textrm{something})$ which is non trivial in $H_0(\chain{D})$, and hence $\disp{\sum_{i\in p}}|\mathfrak{b}_i|\geq \Big|\disp{\sum_{i\in p}}\mathfrak{b}_i\Big|\geq d_{\chain{D}}$.

Now we consider $\Omega_{j_0}$ as in the statement of the lemma. We obtain then
\[
Kd_{{\chain{C}}\otimes{\chain{D}}}=K|x_0|=K\left(\sum_{i=1}^{n_{-1}}|\mathfrak{a}_i|+\sum_{i=1}^{n_0}|\mathfrak{b}_i|+\sum_{i=1}^{n}|\mathfrak{c}_i|\right)\geq K\sum_{i=1}^{n_0}|\mathfrak{b}_i|\geq\sum_{p\in\Omega}\sum_{i\in p}|\mathfrak{b}_i|\geq Nd_{\chain{D}},
\]
and since $d_{{\chain{C}}\otimes{\chain{D}}}$ is an integer, it is bounded below by $\left\lceil\frac{N}{K}d_{\chain{D}}\right\rceil$.
\end{proof}

Using the correspondence established in Proposition \ref{prop:ChainToCSS},
this leads to the following statement, formulated in terms of
classical codes and matrices.

\begin{theo} \label{thm:main}
  Let $\qcode{C}$ be a CSS code defined as a pair of classical
  codes $\code{C}_2\subseteq \code{C}_1^\perp$ given by full rank parity--check matrices.
  Let $g_1, \ldots, g_k\in \code{C}_1^\perp$ and $g_1^*, \ldots, g_k^*\in\code{C}_2^\perp$
  be such that
  $$\code{C}_1^\perp = \code{C}_2 \oplus
  \Span(g_1, \ldots, g_k), \quad 
  \code{C}_2^\perp = \code{C}_1 \oplus
  \Span(g_1^*, \ldots, g_k^*)\quad {\rm and}\quad
  \forall i,j, \langle 
  g_i^*, g_j \rangle = \delta_{ij}.$$
  If, for any $j_0 \in \{1, \ldots , k\}$, there exists $\Omega_{j_0} \subseteq
  g_{j_0}^* + \code{C}_1$ and $\Omega'_{j_0} \subseteq g_{j_0} +
  \code{C}_2$, with $|\Omega_{j_0}|,\ |\Omega'_{j_0}| \geq N$
  and $\overlap (\Omega_{j_0}),\ \overlap (\Omega'_{j_0}) \leq K$. Then,
  for every CSS code $\qcode{D}$, the minimum distance of $\qcode{C}
  \otimes \qcode{D}$ satisfies
  $$
  d_{\qcode{C} \otimes \qcode{D}} \geq \left\lceil \frac N K d_{\qcode{D}}
  \right\rceil \cdot
  $$
\end{theo}

Lemma \ref{lem:Th} implies actually a stronger result since
\begin{enumerate}
\item the matrices $\mathbf{H}_X$ and $\mathbf{H}_Z$ may not have full
  rank as long as the matrices of $\qcode{D}$ do;
\item one may use distinct bases to define the sets $\Omega$'s and $\Omega'$'s.
\end{enumerate}
However, our applications shall use only the statement of Theorem
\ref{thm:main}, and mostly in its one dimensional case, which is even
simpler to state. 
  \begin{cor}\label{cor:cor_of_main}
    Let $\CC$ be a CSS code of dimension $1$ associated to a
    pair of classical codes $\code{C}_2 \subseteq \code{C}_1^\perp$
    given by full rank parity-check matrices. Then,
    if there exists a subset $\Omega \subseteq \code{C}_1^\perp \setminus
    \code{C}_2$
    and a subset $\Omega' \subseteq \code{C}_2^\perp \setminus
    \code{C}_1$ with $|\Omega|,|\Omega'| \geq N$ and
    $\overlap (\Omega),\overlap (\Omega') \leq K$. Then for any
    CSS code $\qcode D$, the minimum distance of
    $\CC \otimes \qcode D$ satisfies
    $$
    d_{\CC \otimes \qcode D} \geq \left\lceil \frac{N}{K} d_{\qcode D} \right\rceil \cdot
    $$
  \end{cor}
  \begin{remarque}
If dealing with a symmetric code described by a symmetric chain
    complex, then it is sufficient to consider the sets $\Omega$
    in Theorem \ref{thm:main} or Corollary \ref{cor:cor_of_main}, 
    since they can be used again as sets $\Omega'$.
  \end{remarque}

%  Local Variables: 
%  mode: latex
%  TeX-master: "OnProductCodes"
%  End: 

%% file: 2.3.DirectConsequences.tex
In this section, we state some direct consequences of Lemma \ref{lem:Th}.

\subsubsection{General lower bounds for the minimum distance of tensor
  products}

\begin{cor}\label{cor:DMinBound}
  If ${\chain{C}}$ is a chain complex satisfying the hypothesis of Lemma \ref{lem:Th}, then $d_{\chain{C}}\geq \left\lceil\frac{N}{K}\right\rceil$.
\end{cor}
\begin{proof}
  Apply Lemma \ref{lem:Th} with the reduced complex
  ${\chain{D}}:=\xymatrix{0\ar[r]&\F\ar[r]&0}$.
\end{proof}

\begin{remarque}
  Applied to the case $\Im(\p_{-1})=0$, Corollary \ref{cor:DMinBound}
  provides also a lower bound for the minimum distance of a classical
  code. In this context, the statement essentially follows from the
  remark that, if an element $c$ of a code has a non trivial scalar
  product with a vector $\omega$, then $c$ has at least one non
  trivial entry shared with $\omega$. Now, with the notation of Section
  \ref{sec:Th} and for $j_0$ such that $g_{j_0}\in c$, $c$ has at least one non trivial entry in common with
  each of the $N$ elements of $\Omega_{j_0}$. But each of these entries can appear at
  most $K$ times. It follows that $c$ has at least $\frac{N}{K}$ non
  trivial entries.
\end{remarque}

\begin{remarque}\label{rk:LaLoose}
One may cherish the hope to provide an LDPC family of CSS codes with
  minimum distances growing faster than the square root of the
  lengths by applying Corollary \ref{cor:DMinBound} to evaluate
  the minimum distance of a CSS code $\qcode C$ as some $\frac{N}{K}$. This is however doomed to fail. Indeed, if stacking the
  elements of $\Omega$ into a matrix and counting, column by column,
  the number $V$ of non trivial entries, we obtain $V\leq n_\qcode
  CK$; but on the other hand, counting $V$ row by row, we obtain
  $V\geq Nd_{\chain C^*}$ since, as noticed in Remark
  \ref{rk:CohomNonTriv}, elements of $\Omega$ are non trivial
  cohomology classes. It follows that $\frac{N}{K}d_{\chain C^*}\leq
  n_\qcode C$. In particular, if $\frac{N}{K}$ is close enough to
  $d_\chain C$, we obtain $d_\chain Cd_{\chain C^*}\lesssim n_\qcode C$
  and hence $d_\qcode C=\min(d_\chain C,d_{\chain
    C^*})\lesssim\sqrt{n_\qcode C}$.
\end{remarque}

\begin{cor}\label{cor:dProd}
  If ${\chain{C}}$ and ${\chain{D}}$ are two chain complexes such that one
  of them is balanced, then
\[
\max(d_{\chain{C}},d_{\chain{D}})\leq
d_{{\chain{C}}\otimes{\chain{D}}}\leq d_{\chain{C}} d_{\chain{D}}.
\]
\end{cor}
  \begin{proof}
    The right hand side is a direct consequence of Proposition
    \ref{prop:TechKunneth}. The left hand one is obtained by applying
    Lemma \ref{lem:Th}
    with $N = K = 1$; this always holds when $\Omega_{j_0}$ is a singleton.
  \end{proof}

 As we shall see in several examples, the upper bound is
 sharp.
On the contrary, the lower bound is not, except in some trivial cases.

\begin{prop}\label{prop:NonSharp}
  If ${\chain{C}}$ and ${\chain{D}}$ are two chain complexes such
  that
  $\chain{C}:\xymatrix{C_{-1}\ar@{^(->}[r]^-{\p_{-1}}&C_0\ar@{->>}[r]^-{\p_0}&C_1}$
  is balanced and $\p_0$ is non zero on every element of the basis of
  $C_0$, then $d_{{\chain{C}}\otimes{\chain{D}}}\geq 2d_{\chain{D}}$.
\end{prop}
\begin{proof}
 With the notation of Section \ref{sec:Th}, we consider, for each
 $j_0\in\Inter{1}{r}$, $\Omega_{j_0}:=\Ker_{j_0}^\cperp$. It can actually be
 described as $\omega_{j_0}+\Ker(\p_0) ^\perp$
 with $\omega_{j_0}$
 any element in $\Ker_{j_0}^\cperp$.\\
Now, for any $i\in\Inter{1}{n_0}$, there is at least one element
$f_i\in\Ker(\p_0)^\perp$ such that $\langle f_i,b_i\rangle=1$, otherwise $b_i$
would be contained in
$\big(\Ker(\p_0) ^\perp\big)^\perp\cong\Ker(\p_0)$ and $\p_0(b_i)$
would be zero.
Furthermore, as a generator, $b_i\in
f_i$, and the map $\big(x\mapsto x+f_i\big)$
induces a bijection between the elements in $\Omega_{j_0}$ which
contain $b_i$ and those who do not. It follows that
$\overlap(\Omega_{j_0})=\frac{1}{2}|\Omega_{j_0}|$ and the statement is proved
using Lemma \ref{lem:Th}.
\end{proof}

\begin{remarque}
  Similar results can be obtained for $q$--ary CSS codes. In this case
  we would get $d_{\chain C \otimes \chain D} \geq \frac{q}{q-1} \max (d_{\chain C}, d_{\chain D})$.
\end{remarque}

\begin{remarque}
  The lower bound in Corollary \ref{cor:dProd} can hence be sharp only
  if $\p_0$ vanishes on some generator. If this generator is not in
  the image of $\p_{-1}$, then $d_\chain{C}=1$ and
  $d_{\chain{C}}d_{\chain{D}}=\max(d_{\chain{C}},d_{\chain{D}})$. If
  it is in the image of $\p_{-1}$, then $\chain{C}$ is the direct sum of a chain complex with
  a (useless) summand
  of the form
  $\xymatrix{\Span\big(b'_i\big)\ar@{^(->>}[r]&\Span(b_i)\ar[r]&0}$,
  where $b'_i$ is the unique preimage of $b_i$; and this summand can be removed
  without altering the minimum distance. 
It follows that the lower bound of Corollary \ref{cor:dProd} is sharp if and only if it is equal to
the upper bound.
\end{remarque}

Corollary \ref{cor:dProd} and Proposition \ref{prop:NonSharp} have the
following consequence for the (reduced) tensor product of CSS codes.
\begin{cor}\label{cor:minor_times_2}
  If $\qcode{C}$ and $\qcode{D}$ are two CSS codes described by
  matrices which have no columns of zeros, then
\[
2\max(d_{\qcode{C}},d_{\qcode{D}})\leq
d_{\qcode{C}\otimes_{r}\qcode{D}}.
\]
\end{cor}

In Section \ref{sec:Steane}
we shall see an example where $d_{{\qcode{C}}\otimes{\qcode{D}}}=\left(2+\frac{1}{3}\right)\max(d_{\qcode{C}},d_{\qcode{D}})< d_{\qcode{C}} d_{\qcode{D}}$.

\begin{remarque}
  For chain complexes, the minimum distance of a tensor product is
  bounded above by the product of the minimum distances of the
  summands but, because of the interplay with dual chain complexes,
  this does not hold anymore for CSS codes. An example is given by the
  Tillich--Zémor
  construction, presented in Section \ref{sec:TZ}.
\end{remarque}

\subsubsection{Parameters of iterated tensor powers}

\begin{cor}\label{cor:ParamTensor}
 Let $\qcode{C}$ be the CSS code associated to ${\chain{C}}:=\xymatrix{\F^{b_1}\ar@{^(->}[r]& \F^a\ar@{->>}[r]&\F^{b_2}}$,
  where $a,b_1,b_2\in\N$.
  If $\chain{C}$ and ${\chain{C}}^*$ satisfy both 
  the hypothesis of Lemma \ref{lem:Th} for the same integers
  $N,K\in\N^*$, then $\big({\qcode{C}}^{\otimes\ell}\big)_{\ell\in\N}$
  is a family of CSS codes with
parameters 
\[
\left\llbracket\ \sim\frac{1}{2}\sqrt{\frac{a+2\sqrt{b_1b_2}}{\pi\ell \sqrt{b_1b_2}}}\left(a+2\sqrt{b_1b_2}\right)^\ell \ ;\ 
(a-b_1-b_2)^\ell \ ;\ \geq \left(\frac{N}{K}\right)^\ell \ ;\ \leq a\ell \ \right\rrbracket.
\]
In particular, the family is logarithmically LDPC and the minimum
distance grows strictly faster than the
$\frac{\log N-\log K}{\log(a+2\sqrt{b_1b_2})}$-th power of the length.
\end{cor}
\begin{proof}
  The statement on
  \begin{itemize}
  \item the length follows from an adaptation of the proofs of Prop. 4.1,
A.1 and A.2 in \cite{Audoux}, details can be found in Appendix \ref{appendix:A};
  \item the dimension is a direct consequence of the Künneth formula;
  \item the bound on the minimum distance follows from an inductive use of
Lemma \ref{lem:Th};
  \item the weight follows from an inductive use of Proposition
    \ref{prop:ProductParameters} and the fact that $b_1,b_2\leq a$,
    which is a consequence of the reducedness of $\chain{C}$.
  \end{itemize}
\vspace{-.4cm}
\end{proof}

Corollary \ref{cor:ParamTensor} can be improved by using the reduced
notion of tensor powers defined in Section \ref{sec:CSSpowers}:
parameters $k$ and $d$ are kept untouched, parameter $w$ is possibly
reduced and parameter $n$ is significantly reduced. To avoid making
the text cumbersome, we only state here the 
case $b=b_1=b_2$, but the
general statement is given in Appendix \ref{appendix:B}.

\begin{cor}\label{cor:ParamTensorReduc}
  Let $\qcode{C}$ be the CSS code associated to
  ${\chain{C}}:=\xymatrix{\F^{b}\ar@{^(->}[r]&
    \F^a\ar@{->>}[r]&\F^{b}}$, where $a,b\in\N$. If $\chain{C}$ and
  ${\chain{C}}^*$ satisfy both the hypothesis of Lemma \ref{lem:Th}
  for the same integers $N,K\in\N^*$, then for every $\ell\in\N$,
  ${\qcode{C}}^{\rotimes\ell}$ is a CSS code with parameters
\[
\left\llbracket\ \frac{2(a+b)^\ell+(a-2b)^\ell}{3}\ ;\  (a-2b)^\ell \
  ;\ \geq \left(\frac{N}{K}\right)^\ell  \ ;\ \leq a\ell \ \right\rrbracket.
\]
This provides a family 
logarithmically LDPC with a minimum distance which grows at least as the $\frac{\log N-\log K}{\log(a+b)}$-th power of the length.
\end{cor}
\begin{proof}
  Compared to Corollary \ref{cor:ParamTensor}, only the statement on
  lengths needs a further proof.
  We set ${\chain{C}}_1:={\chain{C}}$ and define recursively ${\chain{C}}_\ell$
  as the tensor product of ${\chain{C}}$ with the reduction of
  ${\chain{C}}_{\ell-1}$. We define the sequences of integers $(a_\ell)_{\ell\in\N^*}$ and
$(b_\ell)_{\ell\in\N^*}$ by
${\chain{C}}_\ell=:\xymatrix{\F^{b_\ell}\ar@{^(->}[r]&\F^{a_\ell}\ar@{->>}[r]&\F^{b_\ell}}$. Developping
${\chain{C}}_\ell\otimes {\chain{C}}$ and using Künneth formula to say that the
homology is trivial except in degree zero, we obtain the
following --- for simplicity, we have written only the dimensions of the
different spaces --- where the second line corresponds to the result
after reduction:
\[
\vcenter{\hbox{$\xymatrix@!0 @R=1cm @C=2.2cm{
0\ar[r]&bb_\ell\ar[r]&ab_\ell+ba_\ell\ar[r]&aa_\ell+2bb_\ell\ar[r]\ar@{->>}[rd]&ab_\ell+ba_\ell\ar[r]&bb_\ell\ar[r]&0\\
&&ab_\ell+ba_\ell-bb_\ell\ar@{^(->}[ru]&&ab_\ell+ba_\ell-bb_\ell&&
}$}}.
\]
It follows that
$\left(\!\!\begin{array}{c}a_{\ell+1}\\b_{\ell+1}\end{array}\!\!\right)=A
\left(\!\!\begin{array}{c}a_{\ell}\\b_{\ell}\end{array}\!\!\right)$
with
$A=\left(\!\!\begin{array}{cc}a&2b\\b&a-b \end{array}\!\!\right)=\left(\!\!\begin{array}{cc}1&-2\\1&1\end{array}\!\!\right)^{-1}\left(\!\!\begin{array}{cc}a-2b&0\\0&a+b\end{array}\!\!\right)
\left(\!\!\begin{array}{cc}1&-2\\1&1\end{array}\!\!\right)$, and hence
that
$\left(\!\!\begin{array}{c}a_{\ell}\\b_{\ell}\end{array}\!\!\right)=A^{\ell-1}\left(\!\!\begin{array}{c}a\\b\end{array}\!\!\right)=\disp{\frac{1}{3}}\left(\!\!\begin{array}{c}2(a+b)^\ell+(a-2b)^\ell
                                                                                                                                                                 \\(a+b)^\ell-(a-2b)^\ell\end{array}\!\!\right)$.
\end{proof}

\begin{remarque}
  The value $\left(\frac{N}{K}\right)^\ell$ given as a lower bound for
  $d_\ell$ in
  Corollaries \ref{cor:ParamTensor} and \ref{cor:ParamTensorReduc} can
  actually be sharpened into
  $\bigg\lceil\!\!\cdots\!\!\Big\lceil\big\lceil\frac{N}{K}\big\rceil\frac{N}{K}\Big\rceil\frac{N}{K}\cdots\bigg\rceil$. It
  provides, in general, a slightly better constant.
\end{remarque}

Proposition \ref{prop:NonSharp} together
with either Corollary \ref{cor:ParamTensor} or
\ref{cor:ParamTensorReduc} imply that, by considering its tensor
powers, any CSS code, even the poorest one (as soon as it is not
defined with matrices containing columns of zeros),
provides a logarithmically LDPC family with a
minimum distance tending
to infinity:
\begin{cor}\label{cor:NonTrivial}
  If $\qcode{C}=\left(\mathbf{H}_X,\mathbf{H}_Z\right)$ is a CSS code
  such that none of $\mathbf{H}_X$ or $\mathbf{H}_Z$ has a zero column,
  then the families $\left(\qcode{C}^{\otimes\ell}\right)_{\ell\in\N}$
  and $\left(\qcode{C}^{\rotimes\ell}\right)_{\ell\in\N}$ are
  logarithmically LDPC with
  $d_{\qcode{C}^{\otimes\ell}},d_{\qcode{C}^{\rotimes\ell}}\geq
  2^\ell$ for every $\ell\in\N^*$.
\end{cor}

In particular, and even if a CSS code has no quantum degeneracy,
i.e. its quantum minimum distance is not larger than the minimum
of the distances of the two classical codes defining it,
for a large enough $ \ell$, its $\ell$--th iterated power does.

\begin{remarque}
  The previous statement asserts that the row weight of the $\ell$--th
  iterated tensor power $\CC^{\otimes \ell}$ (or $\CC^{\otimes_r
    \ell}$) is linear in $\ell$ and hence
  logarithmic in the code length. Very similar arguments would show
  that the column weights, which are related to the number of
  stabilizers acting non trivially on a given qubit, are also in $O
  (\ell)$ and hence logarithmic in the code length.
\end{remarque}

%  Local Variables: 
%  mode: latex
%  TeX-master: "OnProductCodes"
%  End: 

%% file: 3.1.TillichZemor.tex
In \cite{TZ}, J.-P. Tillich and G. Z\'emor 
give a construction of a CSS
code from any two classical codes. Their construction is based on a
graph point of view. In this section, we give an
alternative approach of their construction based on tensor products.

Any linear map between two $\F$--spaces can be seen as a chain complex
of length 2 and, by adding a null space on the right or on the left,
as a short complex. If interested only in reduced complexes, one can
apply the reduction process described in Section \ref{sec:Reduction}
or, equivalently, consider only injective maps (with a null space on
their right) and surjective map (with a null space on their left). The
following two propositions can be straightforwardly verified.
\begin{prop}
  If $\chain{C}:=\xymatrix{C_{-1}\ar@{^(->}[r]^-{g}&C_0\ar@{->>}[r]&0}$ is given
  with a basis $\BB$, then
  \begin{itemize}
  \item $n_{\chain{C}}=\dim(C_0)$;
  \item $H_0(\chain{C})=\Coker(g)$ so
    $k_{\chain{C}}=\dim(C_0)-\dim(C_{-1})$ and, if $\Coker(g)\ncong0$, $d_{\chain{C}}=1$;
  \item $w_{\chain{C}}=0$.
  \end{itemize}
\end{prop}

\begin{prop}
  If $\chain{C}:=\xymatrix{0\ \ar@{^(->}[r]&C_0\ar@{->>}[r]^(.4){p}&C_1}$ is given
  with a basis $\BB$, then
  \begin{itemize}
  \item $n_{\chain{C}}=\dim(C_0)$;
  \item $H_0(\chain{C})=\Ker(p)$ so
    $k_{\chain{C}}=\dim(C_0)-\dim(C_1)$ and, if $\Ker(p)\ncong0$,
    $d_{\chain{C}}=\pmin_{x\in\Ker(p)\setminus\{0\}}|x|$;
  \item $w_{\chain{C}}$ is the
    maximal weight of a row in $\Mat_\BB(p)$.
  \end{itemize}
\end{prop}

Now, noting that
\[
\left(\xymatrix@!0@C=1.25cm{C_{-1}\ar@{^(->}[r]^-{g}&C_0\ar@{->>}[r]&\ 0}\right)^*=\xymatrix@!0@C=1.25cm{C_{-1}&C_0\ar@{->>}[l]_(.3){g^*}&0\ar@{_(->}[l]}
\textrm{ and }
\left(\xymatrix@!0@C=1.25cm{0\ \ar@{^(->}[r]&C_0\ar@{->>}[r]^(.4){p}&C_1}\right)^*=\xymatrix@!0@C=1.25cm{0&C_0\ar@{->>}[l]&C_1\ar@{_(->}_-{p^*}[l]},
\]
and recalling from Section \ref{sec:ClassicalCodes} that a classical
code $\code{C}$ can be given either by an injective generating map
$g_\code{C}$ or by a surjective parity-check map $p_\code{C}$, and
that, up to transpose, $\code{C}$ and $\code{C}^\perp$ exchange their
generating and parity-check maps, we obtain as a corollary of Proposition \ref{prop:ChainToCSS}:
\begin{prop}
  If $\code{C}$ is a classical code, then the CSS code $\qcode{C}_p$
  associated to $\xymatrix{0\
    \ar@{^(->}[r]&C_0\ar@{->>}[r]^(.4){p_\code{C}}&C_1}$ has
  parameters $n_{\qcode{C}_p}=n_{\code{C}}$,
  $k_{\qcode{C}_p}=k_{\code{C}}$, $d_{\qcode{C}_p}=1$ and
  $w_{\qcode{C}_p}=w_{\code{C}}$; and the CSS code $\qcode{C}_g$
  associated to
  $\xymatrix{C_{-1}\ar@{^(->}[r]^-{g_\code{C}}&C_0\ar@{->>}[r]&0}$ has
  parameters $n_{\qcode{C}_g}=n_{\code{C}^\perp}$,
  $k_{\qcode{C}_g}=k_{\code{C}^\perp}$, $d_{\qcode{C}_g}=1$
  and $w_{\qcode{C}_g}=w_{\code{C}^\perp}$.
\end{prop}

It can hence be noted that CSS codes associated to classical codes
have very poor minimum distances. However, combining Proposition
\ref{prop:ChainToCSS} with Proposition \ref{prop:ProductParameters}
and Corollary \ref{cor:dProd}, we obtain
Tillich--Zemor result which can be stated as:
\begin{theo}\cite{TZ}
  If $\code{C}$ and $\code{D}$ are two classical codes given,
  respectively, by a parity-check map $p_\code{C}$ and a generating
  map $g_\code{D}$, then the CSS code $\qcode{C}\otimes\qcode{D}$
  associated to
  $\left(\xymatrix{C_0\ar@{->>}[r]^(.4){p_\code{C}}&C_1}\right)\otimes\left(\xymatrix{D_{-1}\ar@{^(->}[r]^-{g_\code{D}}&D_0}\right)$
  has parameters
  $n_{{\qcode{C}}\otimes{\qcode{D}}}=n_{\code{C}}n_{\code{D}^\perp}+k_{\code{C}}k_{\code{D}^\perp}$,
  $k_{{\qcode{C}}\otimes{\qcode{D}}}=k_{{\code{C}}}k_{\code{D}^\perp}$,
  $d_{{\qcode{C}}\otimes{\qcode{D}}}=\min(d_{{\code{C}}},d_{\code{D}^\perp})$
  and
  $w_{{\qcode{C}}\otimes{\qcode{D}}}=\max\big(w_{{\code{C}}}+w_{{{\code{D}}}^\perp},w_{{\code{C}}^\perp}+w_{{\code{D}}}\big)$. 
\end{theo}

\begin{remarque}
  {The fact that one has to combine two classical codes described,
    respectively, by a parity-check and a generating matrix should be
    compared to the necessity, in Tillich--Zemor construction, to deal
  with a classical code and the dual of another one. This is pictured
  by the butterfly crossed polygon in the right-hand side of \cite[Figure 5]{TZ}.}
\end{remarque}

%  Local Variables: 
%  mode: latex
%  TeX-master: "OnProductCodes"
%  End: 

%% file: 3.2.Khovanov.tex
Chain complexes arise naturally in the context of topology, and in
particular in the framework of knot and link theory.
Khovanov homology is an example of link invariant which is defined as the homology of a
chain complex $\Ch(D)$ associated to any link
diagram $D$.
In
\cite{Audoux}, the first author used it to define CSS codes associated to link
diagrams.
Khovanov homology is related to tensor products via the following proposition:
\begin{prop}
 For any pointed link diagrams $D_1$ and $D_2$, $\Ch(D_1\#
 D_2)=\Ch(D_1)\otimes\Ch(D_2)$ where $\#$ denotes the pointed connected sum.
\end{prop}

\subsubsection{Unknot codes}

The diagrams used in \cite{Audoux} to define the unknot codes are not
iterated connected sums of a given diagram. However, their Khovanov
chain complexes is isomorphic to that of the following diagrams
\[
\dessin{2.5cm}{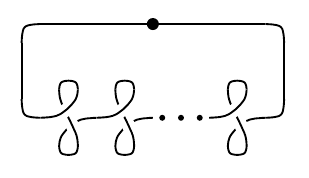}
\]
which are iterated connected sums.
It follows that the chain complexes underlying unknot codes are the $\ell$-th tensor power of
\[
\vcenter{\hbox{$\xymatrix@!0 @R=.35cm @C=2.5cm{
&\bullet\ar@{-}[rddd]&\\
&&\\
&\bullet\ar@{-}[rddd]&\\
\bullet\ar@{-}[ruuu]\ar@{-}[ru]\ar@{-}[rd]&&\bullet\\
&\bullet\ar@{-}[ru]\ar@{-}[rd]&\\
\bullet \ar@{-}[ru]\ar@{-}[rd]\ar@{-}[rddd]&&\bullet\\
&\bullet\ar@{-}[ruuu]&\\
&&\\
&\bullet\ar@{-}[ruuu]&
}$}}.
\]
We denote the generators in the middle degree, from top to bottom, by
positive integers from 1 to 5.
Using the subset
notation described in the Notation section, the homology is generated by
$14$. Since
the chain complex is symmetric, it is sufficient to use Corollary \ref{cor:ParamTensor} with
$\Omega_{14}:=\big\{12,45\big\}$ to obtain
back the parameters
$\param{\frac{3^{2\ell+1}}{\sqrt{8\pi\ell}}}{1}{2^\ell}{3\ell}$.
Using Corollary \ref{cor:ParamTensorReduc}, it can be improved into a
family with asymptotical parameters $\param{\frac{2 \cdot 7^\ell}{3}}{1}{2^\ell}{3\ell}$.

\subsubsection{Unlink codes}

The diagrams considered in \cite{Audoux} to define unlink codes are
iterated connected sums of the following diagram
\[
\dessin{1.95cm}{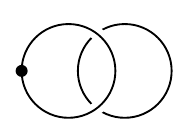}
\]
so the associated chain complexes are iterated tensor powers of
\[
\vcenter{\hbox{$\xymatrix@!0 @R=.35cm @C=2.5cm{
&\bullet\ar@{-}[rddd]&\\
&&\\
&\bullet\ar@{-}[rd]&\\
\bullet\ar@{-}[ruuu]\ar@{-}[ru]\ar@{-}[rd]\ar@{-}[rddd]&&\bullet\\
&\bullet\ar@{-}[ru]&\\
&&\\
&\bullet\ar@{-}[ruuu]&
}$}}.
\]

We denote the generators in the middle degree, from top to bottom, by
positive integers from 1 to 4.
Using the subset
notation, 
the homology is generated by
$12$ and $13$, and using Corollary \ref{cor:ParamTensor} with
$\Omega_{12}:=\big\{24,13\big\}$ and $\Omega_{13}:=\big\{34,12\big\}$, we obtain
back the parameters
$\param{\sqrt{\frac{3}{2\pi\ell}}6^\ell}{2^\ell}{2^\ell}{4\ell}$.
Using Corollary \ref{cor:ParamTensorReduc}, it can be improved into a
family with asymptotical parameters $\param{\frac{2}{3}5^\ell}{2^\ell}{2^\ell}{4\ell}$.

\begin{remarque}
  Forgetting its Khovanov origin, the above family can
  be extended to a two-parameters family defined as the $\ell$-th tensor power
  of
  \[
  \vcenter{\hbox{$\xymatrix@!0 @R=.35cm @C=2.5cm{
        &\bullet\ar@{-}[rddddd]&\\
        &&\\
        &\bullet\ar@{-}[rddd]&\\
        &&\\
        &\raisebox{.1cm}{$\vdots$}&\\
        \bullet\ar@{-}[ruuuuu]\ar@{-}[ruuu]\ar@{-}[rddd]\ar@{-}[rddddd]&2r\textrm{
          generators}&\bullet\\
        &\raisebox{-.2cm}{$\vdots$}&\\
        &&\\
        &\bullet\ar@{-}[ruuu]&\\
        &&\\
        &\bullet\ar@{-}[ruuuuu]& }$}}.
  \]
  From a coding theoretic point of view, the corresponding
CSS code is symmetric and associated to the
code $\code{C} \subseteq \code{C}^{\perp}$ where $\code{C}$
is the repetition code and $\code{C}^{\perp}$ is the parity
code. 
  The homology is generated by $\big\{1i\ |\ i\in\Inter{2}{2r-1}\big\}$ and, using
  Corollary \ref{cor:ParamTensorReduc} with
  $\Omega_{1i}:=\bigg\{i(2r),\pprod_{\substack{j=1\\j\neq i}}^{2r-1}j\bigg\}$,
  we obtain codes with asymptotical parameters
  $\param{\frac{2}{3}(2r+1)^\ell}{(2r-2)^\ell}{2^\ell}{2r\ell}$ when
  $r$ is fixed and $\ell$ tends to infinity. 
\end{remarque}

%  Local Variables: 
%  mode: latex
%  TeX-master: "OnProductCodes"
%  End: 

%% file: 3.3.Steane.tex
In \cite[Section V.A]{Bravyi}, which is the extended version of
\cite{Hastings}, Bravyi and Hastings study in details the Steane
code with parameter $\llbracket7;1;3\rrbracket$. In its principal symmetric form, it
can be described as the CSS code $\qcode{S}_{7;1;3}$ associated to
\[
\vcenter{\hbox{$\xymatrix@!0 @R=.7cm @C=3cm{
&\bullet \ar@{-}[rd]&\\
\bullet \ar@{-}[ru]\ar@{-}[r]\ar@{-}[rd]\ar@{-}[rddd]& \bullet \ar@{-}[r]\ar@{-}[rdd]& \bullet \\
&\bullet \ar@{-}[ru]\ar@{-}[rd]\ar@{-}[rddd]&\\
\bullet \ar@{-}[ruu]\ar@{-}[ru]\ar@{-}[r]\ar@{-}[rdd]& \bullet \ar@{-}[r]& \bullet\\
&\bullet \ar@{-}[ruuu]\ar@{-}[rd]&\\
\bullet \ar@{-}[ruuu]\ar@{-}[ru]\ar@{-}[r]\ar@{-}[rd]& \bullet \ar@{-}[ruu]\ar@{-}[r]& \bullet\\
&\bullet \ar@{-}[ru]&
}$}}.
\]
We denote the generators in the middle degree, from top to bottom, by
positive integers from 1 to 7.
Using the subset notation, it is easily computed that
$\Ker(\p_0)=\F\langle1235,2346,3567,124\rangle$ and $\Im(\p_{-1})=\F\langle1235,2346,3567\rangle$. The homology
is hence generated by $124$.

Bravyi and Hastings computed that
$d_{\qcode{S}_{7;1;3}^{\otimes2}}=7$, and indeed, using Theorem
\ref{lem:Th} with
\[
\Omega_{124}:=\big\{124,136,157,237,256,345,467\big\},
\]
we obtain that $d_{\qcode{S}_{7;1;3}^{\otimes2}}\geq \frac{7}{3}d_{\qcode{S}_{7;1;3}}=7$. This an example where Theorem
\ref{lem:Th} gives a sharp lower bound whereas $K\neq1$ and
$\frac{N}{K}\notin\N$.

We shall see in Section \ref{sec:QFG} a generalization of
$\qcode{S}_{7;1;3}$.

%  Local Variables: 
%  mode: latex
%  TeX-master: "OnProductCodes"
%  End: 

%% file: 3.4.BravyiHastings.tex
In \cite{Hastings}, Bravyi and Hastings present a notion of
homological product for CSS codes which are described by 2--nilpotent
maps. This product is closely related to the tensor product of codes.
\begin{defi}
Let $C, D$ be two $\F$--spaces, and $\p_C\in\End(C),\p_D\in\End(D)$
be two $2$--nilpotent maps\footnote{Recall from Definition
  \ref{def:ChainComplex} that a 2--nilpotent map is an endomorphism
  $\p$ satisfying $\p^2=0$.}.
As recalled in Remark \ref{rem:BH}, these data provide two 
CSS codes $\qcode{C}, \qcode{D}$
and the {\em homological product} $\qcode C \boxtimes \qcode D$
is defined as the 
CSS code associated to the $2$--nilpotent
map
\[
\p_C\boxtimes\p_D:=\p_C\otimes\Id_D+\Id_C\otimes\p_D\in\End(C\otimes
D).
\]
\end{defi}
\begin{prop}[\cite{Hastings}]\label{prop:BH}
  If ${\qcode{C}}$ and ${\qcode{D}}$ are two CSS codes described by
  2--nilpotent maps, then $k_{{\qcode{C}} \boxtimes{\qcode{D}}}=k_{{\qcode{C}}}k_{{\qcode{D}}}$ and $\max(d_{{\qcode{C}}},d_{{\qcode{D}}})\leq d_{{\qcode{C}} \boxtimes{\qcode{D}}}\leq d_{{\qcode{C}}}d_{{\qcode{D}}}$.
\end{prop}
Bravyi and Hastings show moreover that for a random CSS code $\qcode{C}$
of length $n$, the minimum distance of $\qcode{C}^{\boxtimes 2}$ is larger than
$cn^2$ for some positive constant $c$ with a probability tending to $1$
when $n$ tends to infinity.

In the coming section, we explain how Bravyi and Hastings' homological
product can be understood as extracted from the tensor product. It
shall follow that our criterion for a lower bound on the minimal
distance, as well as all its corollaries, apply in the same way to
homological products. From this perspective, homological products
appear as an improvement of tensor products since they reduce the
length of the outputs while preserving the dimension. Minimum
distances are however more difficult to compare, even if they
share a same lower bound.

Conversely, we then show that the situation is
inverted when starting from chain complexes: 
tensor products can be understood as extracted from the homological
products of the associated ungraded 2--nilpotent maps. In this
situation, minimum distances for tensor and homological products are
equal, so the tensor product has globally better relative parameters.

\subsubsection{From 2--nilpotent maps to chain complexes}\label{sec:2nilpo->chain}
Let $\qcode{C}$ be a CSS code described by a 2--nilpotent map
$\p_C\in\End(C)$. From the chain complex point of view, $\qcode{C}$ is
also the CSS code associated to
\[
\chain{C}:=\xymatrix{C\ar[r]^-{\p_C}&C\ar[r]^-{\p_C}&C},
\]
which can be reduced into
\[
\xymatrix{C_-\ar@{^(->}[r]^-{\p_C}&C\ar@{->>}[r]^-{\p_C}&C_+},
\]
where $C=:C_-\oplus \Ker(\p_C)$ and $C_+:=\fract{C}/{C'_+}$ with
$C:=C'_+\oplus \Im(\p_C)$. 
We set $\pi_C:C\to C_+$ the
canonical projection and $\p_C^{-1}$ the inverse map of
$\p_C:C_-\to C_+$.  We set similar notation for $\qcode{D}$, another
CSS code described by a 2--nilpotent map.

The length 3 middle part of $\chain{C}\rotimes\chain{D}$ is equal to
\[
\vcenter{\hbox{$\xymatrix@!0 @R=.7cm @C=3cm{
&C_-\otimes D_+\ar[rd]^-{\p_C\otimes\Id_D}\ar@{}[dd]|\bigoplus&\\
C_-\otimes D\ar[ru]^-{\Id_C\otimes\p_D}\ar[rd]^-{\p_C\otimes\Id_D}\ar@{}[dd]|\bigoplus&&C\otimes D_+\ar@{}[dd]|\bigoplus\\
&C\otimes D\ar[ru]^-{\Id_C\otimes\p_D}\ar[rd]^-{\p_C\otimes\Id_D}\ar@{}[dd]|\bigoplus&\\
C\otimes D_-\ar[ru]^-{\Id_C\otimes\p_D}\ar[rd]^-{\p_C\otimes\Id_D}&&C_+\otimes D\\
&C_+\otimes D_-\ar[ru]^-{\Id_C\otimes\p_D}&
}$}}.
\]
It can be decomposed as the direct sum
$\chain{P}_1\oplus\chain{P}_2\oplus\chain{P}_3$, where 
\begin{itemize}
\item $\chain{P}_1$ is the chain subcomplex defined as $\Span\left(
\vcenter{\hbox{$\xymatrix@!0 @R=.7cm @C=3cm{
&w\otimes x\ar[rd]\ar@{}[dd]|\oplus&\\
u\otimes v\ar[ru]\ar[rd]\ar@{}[dd]|\oplus&&0\ar@{}[dd]|\oplus\\
&\p_C(w)\otimes{\p_D^{-1}}(x)\ar[ru]\ar[rd]\ar@{}[dd]|\oplus&\\
0\ar[ru]\ar[rd]&&0\\
&0\ar[ru]&
}$}}
\right)$,
with $u\otimes v\in C_-\otimes D_-$ and $w\otimes x\in C_-\otimes
D_+$. In other words, $\chain P_1$ is defined as
  $C_-\otimes D_-$ in degree $-1$; as the space spanned by elements of the
  form $\big(w\otimes
  x \big)\oplus \big (\p_C(w)\otimes\p_D^{-1}(x) \big)\oplus0\in \big (C_-\otimes
  D_+\big)\oplus \big (C\otimes D \big)\oplus \big (C_+\otimes
  D_-\big)$ for some $w\in C_-$ and $x\in D_+$, in degree 0; and as zero
  in degree 1;
\vspace{.5cm}
\item $\chain{P}_2$ is the chain subcomplex defined as $\Span\left(\vcenter{\hbox{$\xymatrix@!0 @R=.7cm @C=3cm{
&0\ar[rd]\ar@{}[dd]|\oplus&\\
u_1\otimes v_1\ar[ru]\ar[rd]\ar@{}[dd]|\oplus&&y_1\otimes
z_1\ar@{}[dd]|\oplus\\
&w\otimes x\ar[ru]\ar[rd]\ar@{}[dd]|\oplus&\\
u_2\otimes v_2\ar[ru]\ar[rd]&&y_2\otimes z_2\\
&\p_C(w)\otimes \p_D^{-1}\big(\pi_D(x)\big)\ar[ru]&
}$}}\right)$
with $u_1\otimes v_1\in C_-\otimes\Ker(\p_D)$,
$u_2\otimes v_2\in C\otimes D_-$, $w\otimes x\in C\otimes D$,
$y_1\otimes z_1\in C\otimes D_+$ and $y_2\otimes z_2\in C_+\otimes D'_+$;
\vspace{.5cm}
\item $\chain{P}_3$ is the chain subcomplex defined as $\Span\left(\vcenter{\hbox{$\xymatrix@!0 @R=.7cm @C=3cm{
&0\ar[rd]\ar@{}[dd]|\oplus&\\
0\ar[ru]\ar[rd]\ar@{}[dd]|\oplus&&0\ar@{}[dd]|\oplus\\
&0\ar[ru]\ar[rd]\ar@{}[dd]|\oplus&\\
0\ar[ru]\ar[rd]&&y\otimes z\\
&w\otimes x\ar[ru]&
}$}}\right)$ with $w\otimes x\in C_+\otimes D_-$ and $y\otimes z\in C_+\otimes\Im(\p_D)$.
\end{itemize}

It is easily checked that $H_0(\chain{P}_1)\cong
H_0(\chain{P}_3)\cong\{0\}$ and that $\chain{P}_2$ is isomorphic, as a
chain complex, to 
\[
\xymatrix@!0@C=5cm{
\big (C\otimes D_-\big)\oplus \big (C_-\otimes \Ker(\p_D)\big)\ar[r]^-{\Id_C\otimes\p_D+\p_C\otimes\Id_D}&C\otimes D
\ar[r]^-{\Id_C\otimes\p_D+\p_C\otimes\Id_D}&\fract{C\otimes
  D}/{C'_+\otimes D'_+}
},
\]
which is a partially reduced
form of the chain complex associated to $\p_C\boxtimes\p_D$. 
In degree
0, the isomorphism is nothing but the projection $\psi$ onto the
central summand $C\otimes D$.  As a consequence, we obtain that
$H(\p_C\boxtimes\p_D)\cong H_0(\chain{C}\rotimes\chain{D}) \cong
H_0(\chain{C}\otimes\chain{D}) \cong H_0(\chain{C})\otimes
H_0(\chain{D})\cong H(\p_C)\otimes H(\p_D)$. The homological product
can hence be seen as a subcomplex of the tensor product that contains
all the homology. This provides a substantial reduction of the length,
but the variation of the minimum distance is, again, more difficult to
estimate. However, the criterion for a lower bound given in Theorem
\ref{thm:main} still holds.

\begin{theo}\label{theo:mainBH}
  Let $\qcode{C}$ be a CSS code defined by a $2$--nilpotent map
  $\p_{\qcode{C}}$, and let $g_1, \ldots, g_k\in \Ker(\p)$ and $g_1^*, \ldots, g_k^*\in\Im(\p)^\perp$
  be such that
  \[
\Ker(\p) = \Im(\p) \oplus
  \Span(g_1, \ldots, g_k)
  {\rm ,} \quad
  \Im(\p)^\perp=\Ker(\p)^\perp \oplus
  \Span(g_1^*, \ldots, g_k^*)\quad {\rm and}\quad
  \forall i,j, \langle 
  g_i^*, g_j \rangle = \delta_{ij}.
\]
  If, for any $j_0 \in \{1, \ldots , k\}$, there exist $\Omega_{j_0} \subseteq
  g_{j_0}^* + \Ker(\p)^\perp$ and $\Omega'_{j_0} \subseteq g_{j_0} +
  \Im(\p)$, with $|\Omega_{j_0}|,\ |\Omega'_{j_0}| \geq N$
  and $\overlap (\Omega_{j_0}),\ \overlap (\Omega'_{j_0}) \leq K$. Then,
  for any CSS code $\qcode{D}$ defined by a $2$--nilpotent map, we have 
  $$
  d_{\qcode{C} \boxtimes \qcode{D}} \geq \left\lceil \frac N K d_{\qcode{D}}
  \right\rceil \cdot
  $$
\end{theo}

\begin{proof}
  Let $\chain{C}$ and $\chain{D}$ denote the chain complexes
  underlying $\qcode{C}$ and $\qcode{D}$. 
  Lemma \ref{lem:Th} gives a lower bound for the weight of
  homologically non trivial elements in the kernel of
  $\p_{\chain{C}\otimes\chain{D}}$. In particular, it holds for 
  elements in $\chain{P}_2$, and $\psi$ provides a one-to-one
  correspondence between them and 
  homologically non trivial elements in the kernel of
  $\p_{\chain{C}\boxtimes\chain{D}}$. However, the map $\psi$ does not
  preserve the weight. Nonetheless, using notation from Section
  \ref{sec:Th}, $\psi(x_0)$ is actually equal to
  $\disp{\sum_{j=1}^{n_0}b_i\otimes\mathfrak{b}_i}$ so its weight is
  $\disp{\sum_{j=1}^{n_0}|\mathfrak{b}_i|}$, and this is precisely the
  part of $|x_0|$ which is
  bounded below in the proof of Lemma \ref{lem:Th}.
\end{proof}
\begin{cor}\label{cor:mainCorBH}
    If $\qcode{C}$ and $\qcode{D}$ are CSS codes described by $2$--nilpotent
  matrices which have no columns of zeros, then
\[
2\max(d_{\qcode{C}},d_{\qcode{D}})\leq d_{\qcode{C}\boxtimes\qcode{D}}.
\]
\end{cor}

\subsubsection{From chain complexes to 2--nilpotent maps}\label{sec:chain->nilpo}

Forgetting the grading provides a canonical way to produce a
2--nilpotent map from any chain complex. We explain now how the tensor
product of two chain complexes can be seen as extracted from the
homological product of the associated 2--nilpotent maps.  This
actually corresponds to the case of 2--nilpotent maps given with a
basis such that their matrices are block-subdiagonal.

Given a CSS code $\qcode C$ associated to a chain complex
\[
{\chain{C}} =\xymatrix{\cdots \ar[r] & C_i\ar[r]^-{\p_i}
  &C_{i+1}\ar[r] ^-{\p_{i+1}}& \cdots}
\]
where the $C_i$'s are all $\{0\}$ but finitely many of them,
we can define $C := \poplus_{i \in \Z} C_i$ and $\p_C := \poplus_{i \in \Z} 
\p_i$.
The map $\p_C$ is 2--nilpotent and it is easily checked that
\[
\fract{\Ker (\p_C)}/{\Im (\p_C)} = H_{\bullet}(\chain C) = \poplus_{i \in \Z} H_i (\chain C).
\]
In particular, $\fract{\ker (\p_C)}/{\Im (\p_C)}\cong H_0 (\chain C)$
whenever $\chain{C}$ is balanced.

If $\chain{C}$ and $\chain{D}$ are two reduced complexes,
then $\xymatrix{C\otimes D\ar[r]^-{\p_C\boxtimes\p_D}&C\otimes
D\ar[r]^-{\p_C\boxtimes\p_D}&C\otimes D}$ decomposes into the direct sum
$\poplus_{i\in\Z}\big\{\chain{C}\otimes\chain{D}\big\}_i$, where
$\big\{\chain{C}\otimes\chain{D}\big\}_i$ is the length three
truncature of $\chain{C}\otimes\chain{D}$ centered in degree $i$.
They
all have null homology except for the summand $i=0$ which actually corresponds to the central
part of $\chain{C}\otimes\chain{D}$. Moreover, any basis induced from
bases of $\chain{C}$ and $\chain{D}$ respects this direct sum decomposition. It follows
that $k_{\chain{C}\boxtimes\chain{D}}=k_{\chain{C}\otimes\chain{D}}$
and $d_{\chain{C}\boxtimes\chain{D}}=d_{\chain{C}\otimes\chain{D}}$.
Besides, it is easily checked that
$n_{\chain{C}\boxtimes\chain{D}}=n_{\chain{C}}n_{\chain{D}}$. Consequently,
for
$\chain{C}:=\xymatrix{\F^{b_1}\ar@{^(->}[r]&\F^a\ar@{->>}[r]&\F^{b_2}}$
a reduced complex defining a CSS code $\qcode{C}$,
the iterated powers $\qcode{C}^{\boxtimes\ell}$, $\qcode{C}^{\otimes\ell}$ and
$\qcode{C}^{\rotimes\ell}$ have same dimensions and minimum distances
but differents lengths, which are respectively $(a+b_1+b_2)^\ell$, $O {\left(
    \frac{{(a+2\sqrt{b_1b_2})}^\ell}{\sqrt{\ell}} \right)}$ and $O {\left({(a+\sqrt{b_1b_2})}^\ell\right)}$.

\subsubsection{Comparison between tensor and homological powers}
There are two natural notions of product for CSS codes, namely tensor and homological ones, and we have observed how
to switch from one to the other.
They both generate LDPC families when used iteratively.
It is natural to question whether a construction is better than the
other. The answer is actually negative, and the qualities of the family of
codes obtained by iterated tensor or homological powers depend on the
initial descriptive type of the input codes:
\begin{itemize}
\item {\it if the input code is described by a $2$--nilpotent map},
  then one can see it as coming from a chain complex with repeated
  space and map. In this situation,
the homological powers of the original $2$--nilpotent map provide
shorter codes with same dimensions than the tensor
powers. Moreover, the control of the minimum distances provided by the
present paper is equal for both.
\item {\it if the input code is described by a general complex}, then
one can consider the underlying $2$--nilpotent map by
forgetting the grading. In this situation, the tensor powers of the original chain complex provide
shorter codes with same dimensions and minimum distances, hence better
relative parameters, than the homological powers. 
\end{itemize}
A good philosophy should hence be to stick to the original nature of
the 
inputs and use homological products when dealing with 2--nilpotent maps
and tensor products when dealing with chain complexes.

%  Local Variables: 
%  mode: latex
%  TeX-master: "OnProductCodes"
%  End: 

%% file: 4.0.Intro_new_families.tex
In this section we present new families of CSS codes defined 
as iterated tensor powers of some given CSS code.
They all share
a logarithmic LDPC structure and, for a length $N_\ell$ which
tends to infinity, their minimum distance can be ``as
close as possible to $\sqrt{N_\ell}$'' in the sense that, for all
$\alpha < \frac 1 2$, there is such a family whose minimum distance is larger than 
$N_\ell^{\alpha}$.

To control minimum distances, we use Theorem~\ref{thm:main}
which requires the construction of large sets of cohomologically
non trivial 
vectors $\Omega$ with small overlap. For this sake, 
it is natural to search among codes with many automorphisms. This feature
is indeed shared by our three examples, namely:
\begin{itemize}
  \item codes from finite geometry, endowed with a natural action
   of $\mathbf{PGL}(3, \Fq)$;
  \item cyclic codes, i.e. codes of length $n$ with a natural action
    of the cyclic group of order $n$;
  \item Reed Muller codes, endowed with a natural action of the affine
    group.
\end{itemize}

%%% Local Variables: 
%%% mode: latex
%%% TeX-master: "OnProductCodes"
%%% End: 

%% file: 4.2.QFG.tex
In this section, we set $q = 2^s$ for some positive integer $s$.
The idea relies on using points/lines incidence structures
of affine and projective spaces over finite fields to construct LDPC
CSS codes. It has already been used to
construct classical LDPC codes in \cite{KouLinFossorier} and moderate
density parity check quantum codes in \cite{Farinholt}.

Here, we shall consider two incidence
structures:
\begin{itemize}
\item the point/line incidence structure;
\item the point/affine charts incidence structure.
\end{itemize}

\subsubsection{The projective plane}
The projective plane $\P^2 (\Fq)$ is defined as the set of 
lines of $\Fq^3$ passing through the origin. Let us recall classical
facts of this finite geometry:

\begin{prop}\label{prop:basics_finite_geom}
\begin{enumerate}[(i)]
\item[]
\item The plane contains $q^2 + q + 1$ points and $q^2 + q + 1$ lines.
\item\label{it:basics_incidence}
Every line contains $q+1$ points and every point is contained in
$q+1$ lines.
\item Every two distinct points are contained in a unique line and every two
distinct lines meet at a unique point.
\end{enumerate}  
\end{prop}

Note that each of the above statements express the principle of
duality in projective planes, which swaps point and lines
and reverses inclusions.

\begin{exemple}
  For $q=2$, the projective plane is also called {\em Fano plane}.
  It contains 7 points and 7 lines and the point/line incidence structure
  is usually represented by the picture given in Figure~\ref{fig:fano}
  in which the 6 lines and the circle represent the 7 lines of $\P^2 (\F)$.
\end{exemple}

\begin{figure}[!h]
\[  
\dessin{2.5cm}{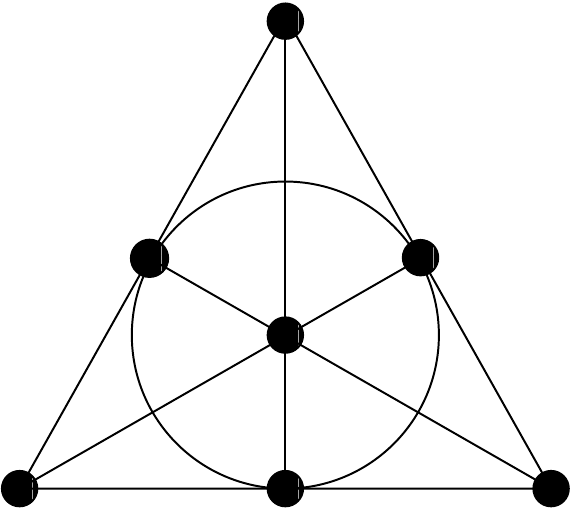}
\]
  \caption{The projective plane $\P^2 (\F)$.}
    \label{fig:fano}
\end{figure}

Additionally we consider the {\em affine charts} of the projective plane.

\begin{defi}
  An affine chart of $\P^2 (\Fq)$ is the complement of a line.
\end{defi}

Let us list some properties of affine charts.

\begin{prop}\label{prop:prop_of_charts}
  \begin{enumerate}[(i)]
  \item[]
  \item An affine chart is isomorphic to the affine plane over $\Fq$;
    in particular it contains $q^2$ elements.
  \item\label{it:charts_orth_lines}
    Let $L$ be a line in $\P^2 (\Fq)$ and $U$ an affine chart. Then,
    \begin{itemize}
    \item either $L$ is the complement of $U$ and hence $L \cap U = \emptyset$;
    \item or $L \cap U$ is an affine line and hence has $q$ elements.
    \end{itemize}
    In particular, since $q$ is even, the number of points of $L \cap U$
    is always even.
  \item The number of affine charts of $\P^2 (\Fq)$ equals the number of lines
    and hence equals $q^2 +q +1$. 
  \end{enumerate}
\end{prop}

\subsubsection{Classical codes associated to projective planes
in characteristic $2$}

We construct two \emph{binary} codes 
associated to the projective space $\P^2 (\Fq)$, with length $|\P^2 (\Fq)| = q^2+q+1$. 
Vectors of $\F^{q^2+q+1}$ can be regarded as subsets of 
$\P^2 (\Fq)$ and we shall freely speak of either vectors
or subsets of $\P^2 (\Fq)$.
From this point of view, the canonical inner product
on $\F^{q^2+q+1}$ can be given a geometric interpretation since, for $S,S' \subseteq \P^2(\Fq)$:
$$
\langle S, S'\rangle = |S \cap S'| \mod 2.
$$

We introduce the codes
\begin{itemize}
\item $\clines(s)$, spanned by lines of $\P^2 (\Fq)$;
\item $\cplanes(s)$, spanned by the affine charts of $\P^2 (\Fq)$.
\end{itemize}

\begin{warning}
  We want to stress the fact that, even though the projective spaces are
  defined over $\Fq$, the associated classical codes, and hence the
  quantum codes to follow, are defined over $\F$.
\end{warning}

The dimension of $\clines (s)$ is well--known.

\begin{prop}[{\cite{Smith}}]
  \label{prop:dim_clines}
  For all $s > 0$, we have $\dim_{\F} \big(\clines (s)\big) = 3^s+1$.
\end{prop}

\begin{prop}
  \label{prop:proj_class_codes}
  For all $s \geq 1$,
  \begin{enumerate}[(i)]
    \item\label{it:cplanes1} $\cplanes (s) \subseteq \cplanes (s)^{\bot}$;
    \item\label{it:cplanes2} $\cplanes (s) \subseteq \clines(s)$;
    \item\label{it:cplanes3} $\cplanes (s) \subseteq \clines (s)^{\bot}$;
    \item\label{it:lines_not_in_cplane} $\clines (s) = \cplanes (s)
      \oplus \Span (L)$ for every line $L \subseteq \P^2 (\Fq)$;
    \item\label{it:dim_cplanes} $\dim \big(\clines (s)\big) - \dim \big(\cplanes (s)\big) = 1$.
  \end{enumerate}
\end{prop}

\begin{proof}
  To prove (\ref{it:cplanes1}), note first that an affine chart has an even
  number of points and hence is orthogonal to itself. Let $A_1, A_2$ be two 
  distinct affine charts. Then, there exist two distinct
  lines $L_1, L_2$ such that
  if we denote by ${^c X}$ the complement of a subset $X$ of $\P^2 (\Fq)$,
  then
  $$
  A_1 = {^c  L_1} \quad {\rm and} \quad 
  A_2 = {^c L_2}.
  $$
  Thus,
  $$
  { A_1 \cap A_2} = {^c ( L_1 \cup L_2)} 
  $$
  Next, since $L_1, L_2$ are distinct to each other,
  $|L_1 \cup L_2| = 2q+1$ and hence 
  $$
  \langle A_1, A_2 \rangle \equiv |A_1 \cap A_2| \equiv |\P^2 (\Fq)| - |L_1 \cup L_2| \equiv q^2 - q \equiv 0
  \mod 2.
  $$
  To prove (\ref{it:cplanes2}), consider an affine chart $A$ and let $L$
  be the line such that $A  = {^c L}$. Let $P \in L$ be a point and $L_1, \ldots , L_q$ be all the lines containing $P$
  but $L$. Then
  $$
  A = L_1 + \cdots + L_q.
  $$
  Indeed, every point $Q \in A$ is in exactly one of the $L_i$'s. 
  Moreover, the $L_i$'s all meet at $P$ which is the only point
  in ${^c A}$ contained in the union of $L_i$'s. Since the number of the
  $L_i$'s is $q$ and hence is even, then $P \notin L_1 + \cdots + L_q$.
  This proves that every affine chart is a sum of lines.

   Point (\ref{it:cplanes3}) is a direct consequence of Proposition~\ref{prop:prop_of_charts}(\ref{it:basics_incidence}).

    To prove (\ref{it:lines_not_in_cplane}), denote by $\one$ the all-one
    vector $(1, \ldots, 1)$. Then for every affine chart $A$, there
    is a line $L$ such that ${^c A} = L$.
    In terms of vectors, we get $A = L + \one$.
    Then, let $L, L'$ be two lines of $\P^2 (\Fq)$ and $A,A'$
    be respectively the affine charts ${^c L}$ and ${^c L'}$, then
    $$
    L + L' = L + L' + \one + \one = A + A'.
    $$
    Thus,
    $$
    L = A+A'+L'.
    $$
    So far, we have proved that every line $L'$ of $\P^2 (\Fq)$ is a sum
    of $L$ and an element of $\cplanes(s)$. This proves that
    $$
    \clines (s) = \cplanes (s) + \Span (L).
    $$
    But $\langle L,L\rangle \equiv |L|
    \equiv 1 \mod 2$, so $L \notin \clines^{\bot}$ and it follows hence,
    from~(\ref{it:cplanes3}), that $L \notin \cplanes (s)$.

    Finally, (\ref{it:dim_cplanes}) is a direct consequence of
    (\ref{it:lines_not_in_cplane}).
\end{proof}

\begin{remarque}
  Actually, $\cplanes (s)$ is nothing but
  the {\em even subcode} of $\clines (s)$ i.e. the subcode of
  vectors of even weight.
\end{remarque}

\subsubsection{Quantum CSS codes from the projective plane in characteristic $2$}
\begin{defi}
  We define $\QFG(s)$ as the quantum code of length $q^2 + q + 1$
  associated to $\cplanes (s) \subseteq \clines (s)$.
\end{defi}

After reduction, the corresponding chain complex is
$$
\chain{C}_\textrm{FG}(s):=\xymatrix{\relax
\F^{3^s} \ar@{^(->}[r] & \F^{2^{2s}+2^s+1} \ar@{->>}[r] & \F^{2^{2s}+2^s-3^s}.
}
$$ 
Indeed, Proposition~\ref{prop:proj_class_codes}(\ref{it:dim_cplanes})
together with Proposition~\ref{prop:dim_clines}
assert that $\dim_{\F} \big(\cplanes (s)\big) = 3^s$.

\begin{remarque}
  The code $\QFG(1)$ is nothing
  but the $\llbracket 7, 1, 3 \rrbracket$ Steane code.
  This fact is actually well--known, since the Steane code
  is known to be constructed from the Hamming code and its dual
  while the Hamming code is already known to be the code $\clines (1)$
  spanned by the lines of $\P^2 (\F)$.
\end{remarque}

\begin{lemme}
  Let $\Omega$ be the set of lines of $\P^2 (\Fq)$. We have
 $$\Omega \subset \clines (s) \setminus \cplanes (s)
 \quad {\rm and}\quad
 \Omega \subset \cplanes(s)^{\bot} \setminus \clines (s)^{\bot}.$$
\end{lemme}

\begin{proof}
  The first inclusion is a direct consequence of
  Proposition~\ref{prop:proj_class_codes}(\ref{it:lines_not_in_cplane}).
  From Proposition~\ref{prop:prop_of_charts}(\ref{it:charts_orth_lines})
  every line of $\P^2 (\Fq)$ is in $\cplanes (s)^{\bot}$. But
  a line $L$ is not in $\clines(s)^{\bot}$. Indeed, let $L'$
  be a line distinct from $L$, then
  $\langle L, L' \rangle \equiv |L \cap L'| \equiv 1 \mod 2$.
\end{proof}

\begin{lemme}\label{lem:QFG_size_Omega}
  $|\Omega| = q^2 +q +1$ and $\overlap(\Omega) = q+1$.
\end{lemme}

\begin{proof}
  It is a direct consequence of Proposition~\ref{prop:basics_finite_geom}.
\end{proof}

\begin{prop}
  For every $s \geq 1$, the family of iterated tensor powers $\QFG(s)^{\otimes \ell}$ has
  parameters
{$$
  \param{\sim\frac{K_s}{\sqrt{\ell}} {\left(\left(2^{2s} + 2^s +1\right) + 2 \left(2\sqrt{3}\right)^s\sqrt{1+\left(\textrm{\scriptsize $\frac{1}{2}$}\right)^s-\left(\textrm{\scriptsize $\frac{3}{4}$}\right)^s}\right)}^\ell}{1}{\geq {\left(2^s+\frac{1}{2^s+1}\right)}^{\ell}}{\leq (2^{2s}+2^s+1)\ell}
  $$}
  for some constant $K_s$ depending only on $s$ and the family of iterated
  reduced tensor powers $\QFG(s)^{\rotimes \ell}$ has parameters
{$$
  \param{\sim K'_s {\left(\left(2^{2s} + 2^s +1\right) + \left(2\sqrt{3}\right)^s\sqrt{1+\left(\textrm{\scriptsize $\frac{1}{2}$}\right)^s-\left(\textrm{\scriptsize $\frac{3}{4}$}\right)^s}\right)}^\ell}{1}{\geq {\left(2^s+\frac{1}{2^s+1}\right)}^{\ell}}{\leq (2^{2s}+2^s+1)\ell}
  $$}
  for some constant $K'_s$ depending only on $s$.
\end{prop}

\begin{proof}
  The minimum distance is a consequence of Lemma~\ref{lem:QFG_size_Omega}
  and Corollary~\ref{cor:cor_of_main}.
  The other parameters are obtained using Corollary~\ref{cor:ParamTensor}
  and Proposition~\ref{prop:length_reduced_asym}.
\end{proof}

%  Local Variables: 
%  mode: latex
%  TeX-master: "OnProductCodes"
%  End: 

%% file: 4.3.QCyclic.tex
The following example is based on classical cyclic codes.

\subsubsection{Cyclic codes}
Here we recall very classical facts about cyclic codes. For further
details we refer the reader to \cite[Chapter 7]{SloaneMcWilliams}.

A binary {\em cyclic code} $\code{C} \subseteq \F^n$ is a code which is stable
under the action of the automorphism
$$\function{\sigma}{\F^n}{(x_0, \ldots, x_{n-1})}{\F^n}{(x_{n-1},x_0 \ldots, x_{n-2}).}$$

\noindent In what follows we identify $\F^n$ and $\fract {\F[X]}/{(X^n-1)}$
using the $\F$--linear isomorphism
$$
\left\{\begin{array}{ccc}
\fract {\F[X]}/{(X^n-1)} &\overset{\sim}{\longrightarrow}& \F^n \\
  f = f_0 + f_1 X + \cdots + f_{n-1}X^{n-1} & \longmapsto &
  (f_0, \ldots,f_{n-1})
\end{array}\right.
$$
and we define the weight of
a polynomial as the number of its nonzero coefficents.
Using this identification, the automorphism $\sigma$ corresponds in
$\fract {\F[X]}/{(X^n-1)}$ to the multiplication by $X$.
A code $\code{C}\subset \fract {\F[X]}/{(X^n-1)}$ is hence cyclic if
it is stable under the multiplication by $X$, that is if it is an ideal. Since $\F [X]$ is a principal ideal ring,
the ideals of $\fract {\F[X]}/{(X^n-1)}$ are in one-to-one correspondence
with the divisors of $X^n-1$.

Given $h \in \F [X]$ such that $h ~|~ X^n-1$, the code $\code{Cyc} (h)$ is defined
as the code corresponding to the ideal generated by $h$. 
It is well--known that this code has dimension $n-\deg (h)$.
The polynomial $h$ is referred to as a {\em generating polynomial}
of the code. It is unique up to multiplication by an invertible element
of $\fract {\F[X]}/{(X^n-1)}$. Note that if $h_1 ~|~ h_2 ~|~ X^n-1$,
then $\code{Cyc}(h_2) \subseteq \code{Cyc}(h_1)$.

The dual of a cyclic code is cyclic and its generating polynomial can
be obtained as follows. Given a polynomial $f \in \fract {\F[X]}/{(X^n-1)}$
we define
$$
\bar{f} := X^{\deg f}f\left(\frac 1 X\right),
$$
and referred to as the {\em reciprocal polynomial} of $f$. Over $\F$, $X^n - 1$
is equal to its reciprocal polynomial so, if $f~|~X^n-1$, then $\bar{f}~|~X^n-1$.
Let $h$ be the polynomial such that $\bar{f}h = X^n-1$, then 
$$
\code{C}(f)^{\bot} = \code{C}(h).
$$

\subsubsection{A construction of CSS codes}
The case $n = 2^s$ is actually never considered in the study of classical cyclic codes since,
in that case, the polynomial $X^n-1$ is completely inseparable and all the
constructions based on choosing divisors of $X^n-1$ having a prescribed set 
of roots, such as BCH codes (see for instance \cite[Chapter 9]{SloaneMcWilliams}),
are not possible.
But oddly enough, this is precisely the case which shall lead to
interesting families
of CSS codes defined by iterated tensor powers.

In this situation $X^{2^s}-1 = (X-1)^{2^s}$ and hence the divisors of $X^n-1$
are of the form $(X-1)^r$ for all $r\in \{0, \ldots , 2^s\}$.
The corresponding cyclic codes are thus $\code{Cyc} \big((X-1)^r\big)$. Such a code
has dimension $n-r$ and 
\[
  {\code{Cyc}\big((X-1)^{r}\big)}^{\bot} = \code{Cyc} \big ((X-1)^{n-r}\big).
\]

\begin{defi}
  For any $r<n$, we define $\QCC(n,r)$ as the CSS code of dimension
  $1$ associated to the pair of codes
  $\code{Cyc}\big((X-1)^r\big) \subseteq
  \code{Cyc}\big((X-1)^{r-1}\big)$.
  If $g_{r} : \F^{n-r} \rightarrow \F^n$ and
  $g_{n-r+1} : \F^{r-1} \rightarrow \F^n$ are, respectively,
  generating maps for $\code{Cyc} \big((X-1)^{r}\big)$ and
  $\code{Cyc}\big((X-1)^{r-1}\big)^{\bot} =
  \code{Cyc}\big((X-1)^{n-r+1}\big)$,
  then $\QCC(n,r)$ is also defined as the CSS code associated to
$$
\chain{C}_\textrm{Cyc}:=\xymatrix{\relax \F^{n-r} \ar@{^(->}[r]^-{g_r}
  & \F^n \ar@{->>}[r]^-{g_{n-r+1}^*} & \F^{r-1}}.
$$
\end{defi}

\begin{lemme}\label{lem:binary_weight}
  Let $r \leq n$ be a non negative integer. Then, the weight of $(X-1)^r \in
  \F[X]$ equals $2^{\textrm{w}_2(r)}$, where $\textrm{w}_2(r)$ 
  denotes the {\em binary weight} of $r$, i.e. the weight of its decomposition
  in base $2$.
\end{lemme}

\begin{proof}
  We prove it by induction on $\textrm{w}_2(r)$. If $\textrm{w}_2 (r) = 1$,
  then $r = 2^a$ for some non negative integer $a$ and $(X-1)^{2^a} = X^{2^a}-1$
  has weight $2$. For $\textrm{w}_2(r) > 1$ then let $a :=
  \lceil \log_2(r)\rceil$. Then $r = 2^a + r'$ where $r' < 2^a$ and
  $\textrm{w}_2(r') = \textrm{w}_2(r)-1 $. By induction, the weight of
  $(X-1)^{r'}$ equals $2^{\textrm{w}_2(r')}$. Hence
  $$
  (X-1)^r = (X-1)^{2^a}(X-1)^{r'} = X^{2^a} (X-1)^{r'} + (X-1)^{r'}. 
  $$
  Since $r' < 2^a$, the polynomials $X^{2^a} (X-1)^{r'}$
  and $(X-1)^{r'}$ have no common monomials and hence, the weight of
  $(X-1)^r$ is twice that of $(X-1)^{r'}$. This concludes the proof.
\end{proof}

\begin{prop}
  For every $\ell\in\N$, the $\ell$--th tensor power of $\QCC(n,r)$ has
  dimension 1 and minimum distance at least $2^{\ell(s- w)}$,
  where $w = \max \big({\rm w}_2(r-1), {\rm w}_2(n-r)\big)$.
\end{prop}

\begin{proof}
To prove the statement, we apply Corollary \ref{cor:cor_of_main} with
\begin{align*}
\Omega & := \left\{X^i (X-1)^{r-1} \mod (X-1)^n \ \big|\ i \in \{0, \ldots, n-1\}\right\} \\
\Omega'& := \left\{X^i (X-1)^{n-r} \mod (X-1)^n \ \big|\ i \in \{0, \ldots, n-1\}\right\}.
\end{align*}
Both sets have cardinality $n = 2^s$, and their respective overlaps are
\[
\overlap (\Omega) = {\rm w}_2(r-1)  \qquad {\rm and} \qquad
\overlap (\Omega') = {\rm w}_2(n-r).
\]
Indeed, stacking the elements of $\Omega$, we obtain a circulant
matrix whose column weight equals the row weight.
But the latter is given by Lemma~\ref{lem:binary_weight}.
This concludes the proof.
\end{proof}

\begin{cor}\label{cor:param_QCycl}
  For $n = 2^s$ and $r = 2^{\frac s 2} = \sqrt{n}$, where $s$ is an
  even integer, the family
  of iterated tensor powers $\QCC(n,r)^{\otimes\ell}$ has parameters
  $$
  \param{\sim\frac{K_n}{\sqrt{\ell}}{\left(n+2n^{\frac{3}{4}}\left(1- \frac{1}{\sqrt{n}}\right)\right)}^\ell}{1}{\geq {\sqrt{n}}^\ell}{\leq \sqrt{n}\ \ell}
  $$
  for some constant $K_n$ depending only on $n$
  and the family
  of iterated reduced tensor powers $\QCC(n,r)^{\otimes\ell}$ has parameters
  $$
  \param{\sim K'_n {\left(n + n^{\frac 3 4}\left(1-\frac{1}{\sqrt{n}}\right) \right)}^\ell}{1}{\geq {\sqrt{n}}^\ell}{\leq \sqrt{n}\ \ell}
  $$
  for some constant $K_n'$ depending only on $n$.
\end{cor}

\begin{proof}
  We have $r-1 = 1+2+\cdots + 2^{\frac s 2 - 1}$
  and $n - r = 2^s - 2^{\frac s 2} = 2^{\frac s 2}(2^{\frac s 2} - 1) = 2^{\frac s 2}
  (1+2+ \cdots + 2^{\frac s 2 -1})$.
  It follows that ${\rm w}_2 (r-1) = {\rm w}_2 (n-r) = \frac s 2$. The
  minimum distance is thus a consequence of
  Corollary~\ref{cor:cor_of_main} and the length a consequence of Corollary~\ref{cor:ParamTensor}
  and Proposition~\ref{prop:length_reduced_asym}.
\end{proof}

%  Local Variables: 
%  mode: latex
%  TeX-master: "OnProductCodes"
%  End: 

%% file: 4.1.QRM.tex
In this section, we define a two-parameters family of CSS codes based
on classical Reed--Muller codes. Such CSS codes have been studied in
\cite{Zhang}. Another construction of stabilizer codes (which are not
CSS) based on Reed--Muller codes was also proposed by Steane in
\cite{SteaneRM}.

To this end, we define, for every $r\in\N^*$ and $s\in\Inter{0}{r}$:
\begin{itemize}
\item $\Pol_r:=\fract{\F[X_1,\ldots,X_r]}/{(X_1^2-X_1,\ldots, X_r^2-X_r)}$ given with
  the basis $\big\{X_I:=\disp{\prod_{i\in I}}X_i\ \big|\ I\subset\Inter{1}{r}\big\}$;
\item $\Pol_{r,s}$ the restriction of $\Pol_r$ to elements of
  degree\footnote{defined by $\deg(X_I)=|I|$, with the convention
    that $0$ has degree $-\infty$} at most $s$;
\item $\phi_r:\Pol_r\hookrightarrow \F^{2^r}$ the map which
  sends a polynomial $P$ to $\big(P(x)\big)_{x\in\F^r}$;
\item $\phi_{r,s}$ the restriction of $\phi_r$ to $\Pol_{r,s}$.
\end{itemize}
\begin{defi}
  The Reed--Muller code $\RM(r,s)$ is the classical code with
  generating map $\phi_{r,s}$.
\end{defi}

\begin{prop}[{\cite[Theorem 13.4]{SloaneMcWilliams}}]\label{prop:RM}
For every $r\in\N^*$ and $s\in\Inter{0,}{r}$,
$\RM(r,s)^\perp=\RM(r,r-s-1)$.
\end{prop}

\begin{defi}
  For every $r\in\N^*$, we define the quantum Reed--Muller code $\QRM(r)$
  as the CSS code associated to
  \[\chain{C}_{\textrm{RM}}(r):=\xymatrix@!0@C=2.5cm{\Pol_{2r,r-1}\ar@{^(->}[r]^-{\phi_{2r,r-1}}&\F^{4^r}\ar@{->>}[r]^-{\phi^*_{2r,r-1}}&\Pol_{2r,r-1}}.
\]
\end{defi}

\begin{prop}
  For every $r\in\N^*$, the family of iterated tensor powers
  $\QRM(r)^{\otimes\ell}$ has parameters
\[
\param{\sim\sqrt{\frac{2^{2r+1}-{2r\choose
        r}}{2\pi\ell\left(4^r-{2r\choose r}\right)}}\left(2^{2r+1}-{2r
      \choose r}\right)^\ell}{{2r\choose r}^\ell}{2^{r\ell}}{\leq 4^r\ell}
\]
and the family of iterated reduced tensor powers $\QRM(r)^{\rotimes\ell}$ has parameters
\[
\param{\frac{2\left(\frac{3.4^r-{2r\choose
          r}}{2}\right)^\ell+{2r\choose r}^\ell}{3}}{{2r\choose
    r}^\ell}{2^{r\ell}}{\leq 4^r\ell}.
\]
\end{prop}
\begin{proof}
  It follows from Proposition \ref{prop:RM} that
  $\Ker(\phi^*_{2r,r-1})=\RM(2r,r)$. But, on the other hand,
  $\Im(\phi_{2r,r-1})=\RM(2r,r-1)$; the homology of $\chain{C}_{\textrm{RM}}(r)$ is
  hence generated by the images through $\phi_{2r}$ of the elements of $\Pol_{2r}$ which are of
  degree exactly $r$.

  Now, let us consider such a generator $\phi_{2r}(X_{I_0})$, with
  $I_0\subset\Inter{1}{2r}$ of cardinality $r$; and set $I^c_0:=\Inter{1}{2r}\setminus I_0$. 
  For every $I,J\subset\Inter{1}{2r}$,
  $\big\langle\phi_{2r}(X_I),\phi_{2r}(X_J)\big\rangle=\disp{\sum_{x\in\F^{4^r}}}X_IX_J(x)=1$
  if and only if $I\cup J=\Inter{1}{2r}$.
  Using notation from Section \ref{sec:Th}, it follows then that 
\[
\Ker^\cperp_{\phi_{2r}(X_{I_0})}=\phi_{2r}\big(X_{I^c_0}\big)+\RM(2r,r)^\perp=\phi_{2r}\big(X_{I^c_0}+\Pol_{2r,r-1}\big).
\]
Since $\QRM(r)$ is symmetric, it is now sufficient to apply Corollary \ref{cor:ParamTensor} or \ref{cor:ParamTensorReduc} with
 \[
\Omega_{\phi_{2r}(X_{I_0})}:=\Big\{\phi_{2r}\Big(\disp{\prod_{i\in I^c_0}(X_i+\e_i)}\Big)\ \Big|\
  \forall i\in I^c_0,\e_i\in\F\Big\}.
\]
 The elements of $\Omega_{\phi_{2r}(X_{I_0})}$ have indeed disjoint support:
  the $x$-th coordinate of $\phi_{2r}\Big(\disp{\prod_{i\in I^c_0}(X_i+\e_i)}\Big)$,
  where $x=(x_i)_{i\in\Inter{1}{2r}}\in\F^{2r}$, is 1 if and only if $x_i=1-\e_i$ for every
  $i\in I^c_0$.

  On the other hand, $\phi_{2r}(X_{\Inter{1}{r}})$ is an element of weight $2^r$
  which survives in homology.
\end{proof}

By extracting the diagonal subfamily $\ell=r$, it follows from
Stirling series that we obtain an
$r$--indexed family with parameters
\[
\param{\sim\frac{\left(\frac{3}{2}\right)^{r-1}4^{r^2}}{e^{\frac{1}{9\pi}}e^{\frac{1}{3}\sqrt{\frac{r}{\pi}}}}}{\sim\frac{4^{r^2}}{e^{\frac{1}{8}}{\sqrt{\pi r}}^r}}{2^{r^2}}{\leq4^rr}.
\]
The family is not, {\it stricto sensu}, logarithmically LDPC, but the
weight grows slower than any positive power of the length, the dimension
faster than any ``$<1$''--power of the length, and the minimum distance
faster than any ``$<\frac{1}{2}$''--power of the length.

%  Local Variables: 
%  mode: latex
%  TeX-master: "OnProductCodes"
%  End: 

%% file: A.LengthPowers.tex
In this appendix, we prove the length part of Corollary
\ref{cor:ParamTensor}. Using inductively Proposition
\ref{prop:ProductParameters}, it is easily seen that the length of
$\qcode{C}^{\otimes\ell}$ is equal to the constant term in
the Laurent polynomial
$(bt^{-1}+a+b't)^\ell$. The statement is hence a consequence of
Proposition \ref{prop:ConstantTerm} given below. But before proving
it, we set some technical lemmata.

\begin{defi}
  For every $x>0$ and every $\e=\pm1$, we define\footnote{in order to
    avoid heavy notation, the
    dependence on $x$ shall be ommited}
  $y_\e=\frac{1+2x+\e\sqrt{1+4x}}{2x}$.\\
  Moreover, for every $\ell\in\N$, we define
  $T_\ell(x):=\psum_{r=0}^\ell{\ell\choose r}{2r\choose r}x^r$ and $P^\e_\ell(x):=x^\ell y_\e^\ell\psum_{r=0}^\ell{\ell\choose
    r}^2y_\e^{-2r}$.
\end{defi}
It is directly checked that $y_++y_-=\frac{1+2x}{x}$,
$y_+-y_-=\frac{\sqrt{1+4x}}{x}$ and $y_+y_-=1$. It follows from the
latter equality that $P_\ell^+(x)=P_\ell^-(x)=x^\ell y_+^\ell\psum_{r=0}^\ell{\ell\choose
    r}^2y_-^{2r}$. Moreover, it can be straightforwardly computed that:
  \begin{lemme}\label{lem:Comb}
    For every $\ell,r\in\Z$,
    \begin{itemize}
    \item $(2\ell+1){\ell\choose r}{2r\choose r}-\ell{\ell-1\choose
        r}{2r\choose r}+2(2\ell+1){\ell\choose r-1}{2r-2\choose
        r-1}-4\ell{\ell-1\choose r-1}{2r-2\choose
        r-1}=(\ell+1){\ell+1\choose r}{2r\choose r}$;\\
    \item $(2\ell+1)\left({\ell\choose r}^2+{\ell\choose
          r-1}^2\right)-\ell\left({\ell-1\choose r}^2-2{\ell-1\choose
          r-1}^2+{\ell-1\choose r-2}^2\right)=(\ell+1){\ell+1\choose r}^2$;
    \end{itemize}
with the convention that ${\ell\choose r}=0$ whenever $\ell < 0$
or $r\notin\Inter{0}{\ell}$.
  \end{lemme}

\begin{lemme}\label{lem:T}
  For every $\ell\in\N$ and $x\geq0$, $T_\ell(x)=P^+_\ell(x)=P^-_\ell(x)$.
\end{lemme}
\begin{proof}
  For every $\ell\in\N^*$
  and every $x>0$,
\[
(\ell+1)T_{\ell+1}(x)-(2\ell+1)(1+2x)T_\ell(x)+\ell(1+4x)T_{\ell-1}(x)=0.
\]
Indeed,
\begin{eqnarray*}
  (2\ell+1)(1+2x)T_\ell(x)
&=&
(2\ell+1)\psum_{r\in\Z}\left({\ell\choose r}{2r\choose r}x^r+2 {\ell\choose
    r}{2r\choose r}x^{r+1}\right);\\
&=&
\psum_{r\in\Z}\left((2\ell+1){\ell\choose r}{2r\choose r}+2 (2\ell+1){\ell\choose
    r-1}{2r-2\choose r-1}\right)x^r;
\end{eqnarray*}
and
\begin{eqnarray*}
  \ell(1+4x)T_{\ell-1}(x)
&=&
\ell\psum_{r\in\Z}\left(
    {\ell-1\choose r}{2r\choose r}x^r+4{\ell-1\choose r}{2r\choose
    r}x^{r+1}\right)\\
&=&
\psum_{r\in\Z}\left(\ell
    {\ell-1\choose r}{2r\choose r}x^r+4\ell{\ell-1\choose r-1}{2r-2\choose
    r-1}\right)x^r\\
\end{eqnarray*}
so, using Lemma \ref{lem:Comb},
\begin{eqnarray*}
  (2\ell+1)(1+2x)T_\ell(x)-\ell(1+4x)T_{\ell-1}(x)
&=&
(\ell+1)\psum_{r\in\Z}{\ell+1\choose r}{2r\choose r}x^r\
    =\ (\ell+1)T_{\ell+1}(x).
\end{eqnarray*}
\noi But on the other hand,
\[
(\ell+1)P^+_{\ell+1}(x)-(2\ell+1)(1+2x)P^+_\ell(x)+\ell(1+4x)P^+_{\ell-1}(x)=0.
\]
Indeed,
\begin{eqnarray*}
  (2\ell+1)(1+2x)P^+_\ell(x)
&=&
(2\ell+1)x(y_++y_-)x^\ell y_+^\ell\psum_{r\in\Z}{\ell\choose
   r}^2y_-^{2r}\\
&=&
x^{\ell+1}y_+^{\ell+1}\psum_{r\in\Z}(2\ell+1){\ell\choose
    r}^2\left(y_-^{2r}+y_-^{2r+2}\right)\\
&=&
x^{\ell+1}y_+^{\ell+1}\psum_{r\in\Z}(2\ell+1)\left({\ell\choose r}^2+{\ell\choose r-1}^2\right) y_-^{2r},
\end{eqnarray*}
and
\begin{eqnarray*}
  \ell(1+4x)P^+_{\ell-1}(x)
&=&
\ell x^2(y_+-y_-)^2x^{\ell-1}y_+^{\ell-1}\psum_{r\in\Z}{\ell-1\choose
    r}^2y_-^{2r}\\
&=&
x^{\ell+1}y_+^{\ell+1}\psum_{r\in\Z}\ell{\ell-1\choose
    r}^2\left(y_-^{2r}-2
    y_-^{2r+2}+y_-^{2r+4}\right)\\
&=&
x^{\ell+1}y_+^{\ell+1}\psum_{r\in\Z}\ell\left({\ell-1\choose r}^2-2{\ell-1\choose r-1}^2+{\ell-1\choose r-2}^2\right)y_-^{2r},
\end{eqnarray*}
so, using Lemma \ref{lem:Comb},
\begin{eqnarray*}
  (2\ell+1)(1+2x)P^+_\ell(x)-\ell(1+4x)P^+_{\ell-1}(x)
&=&
(\ell+1) x^{\ell+1}y_+^{\ell+1}\psum_{r\in\Z}
    {\ell+1\choose r}^2y_-^{2r}\ =\ (\ell+1)P^+_{\ell+1}(x).
\end{eqnarray*}
Finally, $T_0(x)=1=P^+_0(x)$ and $T_1(x)=1+2x=x(y_++y_-)=P^+_1(x)$, so $T_\ell(x)=P^+_\ell(x)$ for every $\ell\in\N$.
\end{proof}

\begin{cor}\label{cor:Tequivalent}
  For every $x\geq0$, when $\ell$ tends to infinity,
\[
T_\ell(x)\ \sim\ \frac{1}{2}\sqrt{\frac{1+4x}{\pi\ell x}}\Big(1+4x\Big)^\ell.
\]
\end{cor}
\begin{proof}
Applying Proposition A.1 in \cite{Audoux} to Lemma \ref{lem:T} (for $\e=1$), we get
\[
T_\ell(x)\ \sim\
x^\ell y_+^\ell \frac{{\left( 1 + \frac{1}{y_+}\right)}^{2\ell +1}}{
2\sqrt{\frac{\pi \ell}{y_+}}}
\ =\ \frac{x^\ell}{2\sqrt{\pi\ell}}\left(\sqrt{y_+}+\frac{1}{\sqrt{y_+}}\right)^{2\ell+1}
\]
But, since $y_+^{-1}=y_-$, $\left(\sqrt{y_+}+\frac{1}{\sqrt{y_+}}\right)^2=y_++2+y_-=\frac{1+4x}{x}$,
and the result follows.
\end{proof}
\begin{prop}\label{prop:ConstantTerm}
  Let $a,b,b'>0$. For every $\ell\in\N$, we denote by $c_\ell$ the constant term in
  $(bt^{-1}+a+b't)^\ell$. Then, when $\ell$ tends to infinity,
\[
c_\ell\sim \frac{1}{2}\sqrt{\frac{a+2\sqrt{bb'}}{\pi\ell \sqrt{bb'}}}\left(a+2\sqrt{bb'}\right)^\ell.
\]
\end{prop}
\begin{proof}
  Following closely the arguments given in the proof of
  Proposition 4.1 in \cite{Audoux}, we begin by
\begin{eqnarray*}
  (bt^{-1}+a+b't)^\ell & = &
                             \Big(\big(\sqrt{b}t^{-\frac{1}{2}}+\sqrt{b'}t^{\frac{1}{2}}\big)^2+a-2\sqrt{bb'}\Big)^\ell\\
  & = & \psum_{r=0}^\ell{\ell\choose
        r}\big(\sqrt{b}t^{-\frac{1}{2}}+\sqrt{b'}t^{\frac{1}{2}}\big)^{2r}\big(a-2\sqrt{bb'}\big)^{\ell-r}\\
  & = & \big(a-2\sqrt{bb'}\big)^\ell\psum_{r=0}^\ell {\ell\choose
        r}\left(\frac{1}{a-2\sqrt{bb'}}\right)^r\psum_{s=0}^{2r}{2r\choose
        s}{\sqrt{b}}^s{\sqrt{b'}}^{2r-s}t^{s-r}\\
& = & \big(a-2\sqrt{bb'}\big)^\ell\psum_{r=0}^\ell\psum_{s=0}^{2r}{\ell\choose
        r}{2r\choose
        s}\left(\frac{b'}{a-2\sqrt{bb'}}\right)^r{\sqrt{\frac{b}{b'}}}^{s}t^{s-r}.
\end{eqnarray*}
In this sum, only the terms with $s=r$ contribute to the constant
term. We obtain hence that
\[
c_\ell\ =\ \big(a-2\sqrt{bb'}\big)^\ell\psum_{r=0}^\ell{\ell\choose
        r}{2r\choose
        r}\left(\frac{\sqrt{bb'}}{a-2\sqrt{bb'}}\right)^r\ =\
      \big(a-2\sqrt{bb'}\big)^\ell T_\ell\left(\frac{\sqrt{bb'}}{a-2\sqrt{bb'}}\right).
\]
Using Lemma \ref{cor:Tequivalent}, we obtain then
\[
c_\ell\ \sim\
\frac{1}{2}\sqrt{\frac{1+\frac{4\sqrt{bb'}}{a-2\sqrt{bb'}}}{\frac{\pi\ell\sqrt{bb'}}{a-2\sqrt{bb'}}}}\left(\left(a-2\sqrt{bb'}\right)\left(1+\frac{4\sqrt{bb'}}{a-2\sqrt{bb'}}\right)\right)^\ell
\ =\ 
\frac{1}{2}\sqrt{\frac{a+2\sqrt{bb'}}{\pi\ell\sqrt{bb'}}}\left(a+2\sqrt{bb'}\right)^\ell.
\]

\end{proof}

%  Local Variables: 
%  mode: latex
%  TeX-master: "OnProductCodes"
%  End: 

%% file: B.LengthRedPowers.tex
Using the same technique as in its proof, Corollary \ref{cor:ParamTensorReduc} can be
generalized to the case $b\neq b'$:

\begin{prop}\label{prop:length_reduced_asym}
  Let $\qcode{C}$ be the CSS code associated to ${\chain{C}}:=\xymatrix{\F^{b}\ar@{^(->}[r]& \F^a\ar@{->>}[r]&\F^{b'}}$,
  where $a,b,b'\in\N$. If $\chain{C}$ and ${\chain{C}}^*$ satisfy both 
  the hypothesis of Lemma \ref{lem:Th} for the same integers
  $N,K\in\N^*$, then for every $\ell\in\N$,
  ${\qcode{C}}^{\rotimes\ell}$ is a CSS code with
parameters $\param{n_\ell}{ (a-b-b')^\ell}{d_\ell}{w_\ell}$, where
$\left(\frac{N}{K}\right)^\ell\leq d_\ell\leq d_1^\ell$,
$w_\ell\leq a\ell$ and
\[
n_\ell=\frac{2bb'\big(a-b-b'\big)^\ell+\Big(b^2+b'^2+(b+b')\sqrt{bb'}\Big)\big(a+\sqrt{bb'}\big)^\ell+\Big(b^2+b'^2-(b+b')\sqrt{bb'}\Big)\big(a-\sqrt{bb'}\big)^\ell}{2\big(b^2+bb'+b'^2\big)}\cdot
\]
This provides a family 
logarithmically LDPC with a minimum distance which grows at least as the $\frac{\log N-\log K}{\log(a+\sqrt{bb'})}$-th power of the length.
\end{prop}
\begin{proof}
    As for Corollary \ref{cor:ParamTensorReduc}, only the statement on the length needs some attention.

We define the sequences of integers $(a_\ell)_{\ell\in\N^*}$,
$(b_\ell)_{\ell\in\N^*}$ and $(b'_\ell)_{\ell\in\N^*}$ by
${\chain{C}}_\ell=:\xymatrix{\F^{b_\ell}\ar@{^(->}[r]&\F^{a_\ell}\ar@{->>}[r]&\F^{b'_\ell}}$. Developping
${\chain{C}}_\ell\otimes {\chain{C}}$ as in the proof of Corollary
\ref{cor:ParamTensorReduc}, we obtain
\[
\left(\!\!\begin{array}{c}b_{\ell+1}\\a_{\ell+1}\\b'_{\ell+1}\end{array}\!\!\right)=
\left(\!\!\begin{array}{ccc}a-b&b&0\\b'&a&b\\0&b'&a-b'\end{array}\!\!\right)
\left(\!\!\begin{array}{c}b_{\ell} \\a_{\ell}\\b'_{\ell}\end{array}\!\!\right).
\]
It can be computed that
\[
\left(\!\!\begin{array}{ccc}a-b&b&0\\b'&a&b\\0&b'&a-b'\end{array}\!\!\right)=
P \left(\!\!\begin{array}{ccc}a-\sqrt{bb'}&0&0\\0&a-b-b'&0\\0&0&a+\sqrt{bb'}\end{array}\!\!\right)P^{-1}
\]
with
\[
P=\left(\!\!\begin{array}{ccc}-\sqrt{b}&b^2&\sqrt{b}\\\sqrt{b'}-\sqrt{b}&-bb'&\sqrt{b'}+\sqrt{b}\\\sqrt{b'}&b'^2&\sqrt{b'}\end{array}\!\!\right).
\]
The result follows by considering the $\ell$-th power.
\end{proof}
%  Local Variables: 
%  mode: latex
%  TeX-master: "OnProductCodes"
%  End: 